\documentclass[a4paper,UKenglish,cleveref, autoref, thm-restate,numberwithinsect]{lipics-v2021}

\title{Maximum Reachability Orientation of Mixed Graphs}
\author{Florian Hörsch}{CISPA, Saarbrücken, Germany \and \url{https://florianhoersch.wordpress.com/} }{florian.hoersch@cispa.d}{https://orcid.org/0000-0002-5410-613X}{}

\authorrunning{F. Hörsch} 

\Copyright{F. Hörsch} 



\category{} 

\relatedversion{} 

\acknowledgements{I want to thank Charupriya Sharma for making me aware of the connection to causality.}

\nolinenumbers 

\EventEditors{John Q. Open and Joan R. Access}
\EventNoEds{2}
\EventLongTitle{42nd Conference on Very Important Topics (CVIT 2016)}
\EventShortTitle{CVIT 2016}
\EventAcronym{CVIT}
\EventYear{2016}
\EventDate{December 24--27, 2016}
\EventLocation{Little Whinging, United Kingdom}
\EventLogo{}
\SeriesVolume{42}
\ArticleNo{23}

\usepackage{amsmath,amsthm}
\usepackage{comment}
\usepackage{amssymb}
\usepackage{graphicx}
\usepackage{thmtools} 
\usepackage{thm-restate}
\usepackage{cleveref}

\newtheorem{Case}{Case}
\newtheorem{question}{Question}

\begin{document}
\maketitle

\begin{abstract}
We aim to find orientations of mixed graphs optimizing the total reachability, a problem that has applications in causality and biology. For given a digraph $D$, we use $P(D)$ for the set of ordered pairs of distinct vertices in $V(D)$ and we define $\kappa_D:P(D)\rightarrow \{0,1\}$ by $\kappa_D(u,v)=1$ if $v$ is reachable from $u$ in $D$, and $\kappa_D(u,v)=0$, otherwise. We use $R(D)=\sum_{(u,v)\in P(D)}\kappa_D(u,v)$.

Now, given a mixed graph $G$, we aim to find an orientation $\vec{G}$ of $G$ that maximizes $R(\vec{G})$. Hakimi, Schmeichel, and Young proved that the problem can be solved in polynomial time when restricted to undirected inputs. They inquired about the complexity in mixed graphs.

We answer this question by showing that this problem is NP-hard, and, moreover, APX-hard.

We then develop a finer understanding of how quickly the problem becomes difficult when going from undirected to mixed graphs. To this end, we consider the parameterized complexity of the problem with respect to the number $k$ of preoriented arcs of $G$, a poorly understood form of parameterization.

We show that the problem can be solved in time $n^{O(k)}$ and that a $(1-\epsilon)$-approximation can be computed in time $f(k,\epsilon)n^{O(1)}$ for any $\epsilon > 0$. 
\end{abstract}

\section{Introduction}
This article deals with the problem of finding an orientation $\vec{G}$ of a mixed graph $G$ such that the total number of ordered pairs $(u,v)$ of vertices in $V(G)$ such that $v$ is reachable from $u$ in $\vec{G}$ is maximized.

This problem has numerous applications in various sciences. For example, in the context of causality, the vertices of the mixed graph may represent random variables and the edges and the arcs of the mixed graph may represent correlations between these random variables. An arc means that the random variable associated to the tail of the arc is causal for the random variable associated to the head of the arc while an edge means that a correlation between the two random variables can be observed, but it is not known which of them is causal for the other one. Hence a solution of the orientation problem helps to better understand the possible correlations between these random variables. More details on the connection of causal relationships and mixed graphs can be found in a survey of Vowels, Camgoz, and Bowden \cite{10.1145/3527154}.

One more concrete, very similar application for this problem comes from biology, namely so-called protein-protein interaction. This application has been a driving force for the research on the problem setting considered in the present article. It has fostered significant interaction between the communities of theoretical computer science and biology, see \cite{doi:10.1089/cmb.2013.0064,inproceedings1,10.1007/978-3-540-87361-7_19}. Concretely, an interaction between certain pairs of proteins can be observed, but it is often technically more difficult to determine which of the two proteins causes this interaction. It is not difficult to see that the collection of protein interactions can be modelled by a mixed graph again. In Life Sciences, it is important to understand which local interactions of proteins lead to most global interaction, which corresponds to our orientation problem. A better understanding of the current problem may hence help to better understand protein-protein interaction, which may be of interest for example in medicine.
\paragraph*{Problem definition}
We now introduce our problem more formally. Basic notation can be found in \Cref{sec2}. In a mixed graph graph $G$, we say that $v$ is reachable from $u$ for some $u,v\in V(G)$ if $v$ can be reached from $u$ by traversing arcs in the given direction and traversing edges in an arbitrary direction.
We define $P(G)$ to be the set of ordered pairs of distinct vertices of $V(G)$. We further define the function $\kappa_G:P(G)\rightarrow \{0,1\}$ by $\kappa_G(u,v)=1$ if $v$ is reachable from $u$ in $G$ and $\kappa_G(u,v)=0$, otherwise, for every $(u,v)\in P(G)$. Given a set $P \subseteq P(G)$, we use $R_P(G)$ for $\sum_{(u,v)\in P}\kappa_G(u,v)$. Observe that digraphs are in particular mixed graphs, so the above definitions carry over to digraphs. An orientation of a mixed graph is obtained by replacing every edge by an arc with the same two endvertices. In the problem Fixed Pairs Mixed Maximum Reachability Orientation (FPMMRO), the input consists of a mixed graph $G$ and a set $P\subseteq P(G)$, and the objective is to find an orientation $\vec{G}$ of $G$ that maximizes $R_P(\vec{G})$. We further use Fixed Pairs Undirected Maximum Reachability Orientation (FPUMRO) for the restriction of FPMMRO to undirected inputs. In the following, when speaking about complexity, we need to carefully distinguish between the version where we are searching for a complete solution of the instance and the maximization version. In the complete solution version, we want to decide whether there exists an orientation $\vec{G}$ of $G$ with $R_P(\vec{G})=|P|$, while in the maximization version, we assume that the instance comes with an integer $K$ and we want to decide whether there exists an orientation $\vec{G}$ of $G$ with $R_P(\vec{G})\geq K$. Clearly, in all settings, the maximization version is at least as hard as the complete solution version.
\paragraph*{Previous work}
For FPUMRO, a simple argument, which was given by Hassin and Meggido \cite{HASSIN1989589}, shows that the complete solution version can be handled in polynomial time. However, it turns out that the maximization version of FPUMRO is significantly harder. More concretely, it was proved by Elberfeld et al. \cite{Elberfeld28112011} that the problem is NP-hard. One important paradigm when encountering NP-hard problems is approximation algorithms. In \cite{Elberfeld28112011}, it was proved that FPUMRO on $n$-vertex graphs can be approximated within a factor of $\Omega(\frac{\log \log n}{\log n})$ in polynomial time, improving on an earlier result of Medevedovsky et al. \cite{10.1007/978-3-540-87361-7_19}. On the other hand, the authors of \cite{Elberfeld28112011} show that, unless P=NP, no polynomial-time approximation within a factor of less than $\frac{12}{13}$ is available. Another variation of FPUMRO taking into account the lengths of the created directed paths was studied by Blokh, Segev, and Sharan \cite{doi:10.1089/cmb.2013.0064}. The study of FPUMRO in the context of parameterized complexity was initialized by Dorn et al \cite{inproceedings1}. They studied several problem-specific parameters describing how much the desired paths connecting input pairs interact. Following up on their work, Cygan, Kortsarz, and Nutov \cite{cygan_kortsarz_nutov_2013} proved that the maximization version of FPUMRO is FPT with respect to the objective function, that is, the number of pairs that are supposed to be satisfied. Interestingly, the existence of a polynomial-time constant-factor approximation algorithm still seems to be open for FPUMRO. 

We now turn our attention to FPMMRO, the version which is more general and hence also more powerful for modelling the applications. 
FPMMRO was first considered by Arkin and Hassin \cite{ARKIN2002271}. They show that, in contrast to FPUMRO, the complete solution version of FPMMRO is already NP-hard in general, but can be solved in polynomial time if $|P|=2$. There exists a rich literature for FPMMRO regarding approximation algorithms and parameterized complexity.

For the approximation side, in \cite{ELBERFELD201396}, Elberfeld et al. showed that, unless $P=NP$, FPMMRO cannot be approximated in polynomial time within a factor better than $\frac{7}{8}$, and, on the other hand, gave a first polynomial-time sublinear approximation algorithm. The approximation ratio was later improved by Gamzu and Medina in \cite{inproceedings}. However, a constant-factor approximation algorithm for FPMMRO was proven to be unlikely by W\l{}odarczyk \cite{wlodarczyk:LIPIcs.ICALP.2020.104}, even when allowing algorithms that run in FPT time with respect to the number of pairs.

For parameterized complexity, the most studied parameter is the number of pairs. It was proved by Cygan, Kortsarz and Nutov in \cite{cygan_kortsarz_nutov_2013} that the complete solution version of FPMMRO is XP with respect to the number of pairs, improving on the result of Arkin and Hassin. 
On the other hand, it was proved by Pilipczuk and Wahlström \cite{10.1145/3201775} that the problem is $W[1]$-hard, which is also implied by the above mentioned result of W\l{}odarczyk. In a recent work, Hanaka et al \cite{hanaka2025structuralparameterssteinerorientation} dealt with the parameterized complexity of FPMMRO when considering other parameters, for example the number of arcs of $G$ or the treewidth of its underlying graph.

We now come to the concrete special case of FPMMRO that is dealt with in the present article. A lot of research on FPMMRO so far focused on the case that the set $P$ of pairs in the instance is small. For example, the above mentioned parameterized complexity results all fall into this category. 
Here, we deal with a specific, but very natural restriction of FPMMRO, which is in some way the other extreme, where $P$ is as large as possible. Namely, we consider the case that the set of pairs in consideration is the set of all possible pairs, i.e. $P=P(G)$.  For simplicity, we use MMRO to denote this restriction and, given a mixed graph $G$, we use $R(G)$ for $R_{P(G)}(G)$. For the complete solution version of MMRO, which corresponds to the question of whether a given mixed graph has a strongly connected orientation, a full characterization directly yielding a polynomial-time algorithm was given by Boesch and Tindell \cite{Boesch01111980}. We hence focus on the maximization version.

In \cite{Hakimi1997OrientingGT}, Hakimi, Schmeichel, and Young considered a setting which is in some way the intersection of MMRO and FPUMRO, namely the special case of MMRO that the input graph $G$ is undirected. They gave a complete characterization of the optimal orientations in this case. A polynomial-time algorithm for the maximization version of this restriction of MMRO follows directly from this characterization.

\paragraph*{Our contributions}
In \cite{Hakimi1997OrientingGT}, the authors also inquired about the complexity of MMRO in general. Surprisingly, despite this question being almost 30 years old, little progress seems to have been made. The first contribution of the present article is to give a negative answer to this question by providing an NP-hardness result. Moreover, the following result is slightly stronger and also shows that the approximability of MMRO has certain limits.
\begin{restatable}{theorem}{mainone}\label{main1}
    Unless P=NP, there exists no polynomial-time $\frac{581773}{581774}$-approximation algorithm for MMRO.
\end{restatable}

While this inapproximability result may appear marginal, it rules out a polynomial-time approximation scheme, which distinguishes MMRO from a large collection of combinatorial problems. With \Cref{main1} at hand, it appears interesting to understand special cases and restrictions of MMRO that are tractable. Parameterized complexity is a powerful paradigm to obtain a better understanding of the complexity of NP-hard problems. In the current case, we wish to consider a parameterization which is natural in the light of the concrete input. Namely, we consider the parameterization with respect to the number of arcs of $G$. On a high level, we wish to understand the complexity of the problem in the case that the number of arcs is small. In \cite{GUTIN2017107}, Gutin et al. studied this parameter for a version of the Chinese Postman Problem in mixed graphs, answering a question of van Bevern et al., see \cite{10.5555/2815545}. We are not aware of any previous consideration of this parameter for orientation problems, except for the work of Hanaka et al. \cite{hanaka2025structuralparameterssteinerorientation}, which was recently submitted independently of the present article. Before describing our results, we wish to remark that for the natural analogous parameterization by the number of edges of $G$, a trivial FPT algorithm exists. Namely, we can compute all possible orientations of $G$ and compare their respective objective values.

In the following, given an instance $G$ of MMRO, we use $k$ for $|A(G)|$. We further say that an orientation $\vec{G}$ of $G$ is {\it optimal} if $R(\vec{G})$ is maximum among all orientations of $G$. Our first result is that, for fixed $k$, we can indeed compute an optimal orientation for an instance of MMRO in polynomial time. More precisely, we show the following result.

\begin{theorem}\label{main2}
    Given an instance $G$ of MMRO, we can compute an optimal orientation of $G$ in time $n^{O(k)}$.
\end{theorem}

\Cref{main2} raises the question whether a further improvement in the running time of this algorithm is possible. Namely, it would be desirable to determine whether the problem is in FPT, that is, if there exists an algorithm computing an optimal orientation that runs in time $f(k)n^{O(1)}$. While this question remains open, we make some progress in this direction by combining the paradigms of parameterized complexity and approximation. To this end, we say that, for some $\alpha \in [0,1]$, an orientation $\vec{G}$ of a given mixed graph $G$ is {\it $\alpha$-optimal} if $R(\vec{G})\geq \alpha R(\vec{G}_0)$ holds for every orientation $\vec{G}_0$ of $G$. The following result roughly speaking shows that the problem is 'almost' FPT, that is, it can be arbitrarily well approximated in FPT time.

\begin{theorem}\label{main3}
    Given a constant $\epsilon >0$ and an instance $G$ of MMRO, in time $f(k,\epsilon)n^{O(1)}$, we can compute a $(1-\epsilon)$-optimal orientation of $G$.
\end{theorem}

Due to Theorems \ref{main1} and \ref{main3}, it becomes evident that for MMRO, algorithms running in FPT time with respect to $k$ can yield solutions of a higher quality than polynomial time algorithms. Indeed, FPT approximations can yield arbitrarily good approximations by \Cref{main3}, which polynomial time algorithms cannot by \Cref{main1}.  

The remainder of this article is structured as follows: in Section \ref{sec2}, we give some more formal definitions and show that the problem can be reduced to mixed graphs that are acyclic. In Section \ref{sec3}, we prove \Cref{main1}. In Section \ref{sec4}, we prove our algorithmic results, that is, \Cref{main2} and \Cref{main3}. Sections \ref{sec3} and \ref{sec4} also contain overviews of the proofs of \Cref{main1} and Theorems \ref{main2} and \ref{main3}, respectively. Finally, in Section \ref{sec5}, we conclude our work.
\section{Preliminaries}\label{sec2}
In this section, we give some preliminaries, we need for our main proofs. In \Cref{not}, we collect the notation we need. In \Cref{tdfzhijo}, we introduce a satisfiability problem, we need for our reduction in \Cref{sec3}. Finally, in \Cref{equi}, we prove the equivalence of MMRO and a closely related problem in acyclic mixed graphs.
\subsection{Notation}\label{not}
We here give a collection of definitions, some of which are repeated from the introduction.
\subsubsection{General basics}
For a positive integer $q$, we use $[q]$ for $\{1,\ldots,q\}$ and $[q]_0$ for $\{0,\ldots,q\}$. For two reals $a,b$ with $a \leq b$, we use $[a,b]$ for the set of reals $c$ satisfying $a \leq c \leq b$. Given a set $S$ of integers, we use $\sum S$ for $\sum_{s \in S}s$. When we say that a certain value equals $f(k)$ for a certain parameter $k$, we mean that there exists a computable function $f$ such the value is bounded by $f(k)$. For some nonnegative integer $\ell$, we use ${\ell \choose 2}$ for $\frac{1}{2}\ell (\ell-1)$. Observe that ${1 \choose 2}={0 \choose 2}=0$. We use $\mathbb{Z}_{\geq 0}$ for the set of nonnegative integers. Given sets $S_1,S_2$, and $Y$ with $S_2\subseteq S_1$ and a function $g:S_1 \rightarrow Y$, we use $g|S_2$ for the restriction of $g$ to $S_2$. If $Y=\mathbb{R}$, we use $|g|$ for $\sum_{s \in S_1}g(s)$. 
\subsubsection{Standard graph theory}
A mixed graph $G$ consists of a vertex set $V(G)$, an edge multiset $E(G)$ and an arc multiset $A(G)$. Here, every $e \in E(G)$ is a subset of $V(G)$ of size exactly $2$ and every $a \in A(G)$ is an ordered pair $(u,v)$ of two distinct vertices $u,v \in V(G)$. We say that $u$ is the {\it tail} of $a$ and $v$ is the  {\it head} of $a$. We write $uv$ for an edge $\{u,v\}\in E(G)$ or an arc $(u,v)\in A(G)$.  A mixed graph $G$ is called an {\it undirected graph} if $A(G)=\emptyset$  and a {\it digraph} if $E(G)=\emptyset$. For $uv \in E(G)\cup A(G)$, we say that $u$ and $v$ are the {\it endvertices} of $uv$. For some $F \subseteq E(G)\cup A(G)$, we use $V(F)$ for the set of vertices that are an endvertex of at least one edge or arc of $F$.

Next, we use $\delta_G(X)$ for the edges in $E(G)$ that have exactly one endvertex in $X$, we use $\delta_G^+(X)$ for the set of arcs in $A(G)$, whose tail is in $X$ and whose head is in $V(G)\setminus X$, and we use $\delta_G^-(X)$ for $\delta_G^+(V(G)\setminus X)$. We say that an arc in $\delta_G^-(X)$ or an edge in $\delta_G(X)$ {\it enters} $X$. Further an edge or an arc {\it leaves} $X$ if it enters $V(G)\setminus X$. Next, we use $d_G(X)$ for $|\delta_G(X)|, d^+_G(X)$ for $|\delta^+_G(X)|$, and $d_G^-(X)$ for $|\delta_G^-(X)|$. For a single vertex $v$, we use $\delta_G(v)$ instead of $\delta_G(\{v\})$ etc..

{\it Orienting} an edge of $G$ means replacing it by an arc with the same two endvertices.  Given a set $F\subseteq E(G)$, an orientation of $F$ is a set of arcs obtained from $F$ by orienting every element of $F$. If a mixed graph $G'$ can be obtained from $G$ by orienting some edges, then we say that $G'$ is a {\it partial orientation} of $G$. Moreover, if $G'$ is a digraph, we say that $G'$ is an {\it orientation} of $G$. Given a mixed graph $G$, its {\it underlying graph $UG(G)$} is obtained by replacing every arc in $A(G)$ by an edge with the same two endvertices. Given a set $X\subseteq V(G)$, we say that that a mixed graph $H$ is obtained from $G$ by {\it contracting $X$} if $H$ is obtained from $G$ by deleting $X$, adding a new vertex $v_X$, adding an edge (arc) $uv_X$ for every edge (arc) $uv$ of $G$ with $u \in V(G)\setminus X$ and $v \in X$ and adding an arc $v_Xu$ for every arc $vu$ of $G$ with $u \in V(G)\setminus X$ and $v \in X$.
Given a mixed graph $G$, an orientation $\vec{G}$ of $G$ and a mixed subgraph $H$ of $G$, we say that the orientation of $H$ {\it inherited} from $\vec{G}$ is the unique orientation of $H$ in which every edge in $E(H)$ has the same orientation it has in $\vec{G}$.

A {\it source} in $G$ is  a vertex $v \in V(G)$ with $d_G(v)=d_G^-(v)=0$ and a {\it sink} in $G$ is  a vertex $v \in V(G)$ with $d_G(v)=d_G^+(v)=0$. A set of two arcs with the same endvertices and different tails is called a {\it digon}.

Given two mixed graphs $G_1$ and $G_2$, we say that $G_2$ is a {\it mixed subgraph} of $G_1$ if $V(G_2)\subseteq V(G_1), E(G_2)\subseteq E(G_1)$, and $A(G_2)\subseteq A(G_1)$ hold. When we say that $G_1$ {\it contains} $G_2$, we mean that $G_2$ is a mixed subgraph of $G_1$.
For $X \subseteq V(G)$, we use $G[X]$ for the mixed subgraph of $G$ induced by $X$, that is, we set $V(G[X])=X$ and we let $E(G[X])(A(G[X]))$ contain all edges (arcs) of $G$ that have both endvertices in $X$.
 
 A directed path is a digraph $P$ on a vertex set $\{v_1,\ldots,v_q\}$ and whose arc set consists of $v_{i}v_{i+1}$ for $i \in [q-1]$. We say that $P$ goes {\it from $v_1$ to $v_q$} and that $P$ is a {\it directed $v_1v_q$-path}.  A {\it mixed path} is a  mixed graph $P$ that can be oriented as a directed path. Moreover, if this directed path can be chosen to go from $u$ to $v$ for some vertices $u,v \in V(P)$, then we say that $P$ goes {\it from $u$ to $v$} and is a {\it mixed $uv$-path.} Given a mixed graph $G$ and two vertices $u,v\in V(G)$, we say that $v$ is {\it reachable} from $u$ in $G$ if $G$ contains a mixed path from $u$ to $v$. A mixed path that is an undirected graph is called a {\it path}.

A {\it directed} walk in a digraph $D$ is a sequence of vertices $(v_1,\ldots,v_q)$ such that $v_{i}v_{i+1}\in A(D)$ for $i \in [q-1]$. We also speak about a {\it directed $v_1v_q$-walk}. It is well-known that a digraph contains a directed $uv$-walk for two vertices $u,v \in V(D)$ if and only if it contains a directed $uv$-path. Given a digraph $D$, $x,y,z\in V(D)$, an $xy$-walk $Q_1=(x=u_1,\ldots,y=u_q)$, and a $yz$-walk $Q_2=(y=v_1,\ldots,z=v_{q'})$, we call the $xz$-walk $(x,u_2,\ldots,u_{q-1},y,v_2,\ldots,v_{q'-1},z)$ the {\it concatenation} of $Q_1$ and $Q_2$. When we speak about the {\it concatenation} of two directed paths, we mean the concatenation of the corresponding walks.

A digraph $D$ is called {\it strongly connected} if for all distinct $u,v \in V(D)$, we have that $v$ is reachable from $u$ in $D$. A {\it directed cycle} is a strongly connected directed graph $C$ such that $d_C^+(v)=d_C^-(v)=1$ holds for all $v \in V(C)$. A {\it mixed cycle} is a mixed graph that can be oriented as a directed cycle. A mixed graph that does not contain a mixed cycle is called {\it acyclic}.  An undirected graph $G$ is {\it connected} if $d_G(X)\geq 1$ holds for every nonempty, proper $X \subseteq V(G)$. A maximal connected subgraph of $G$ is called a {\it component} of $G$.
 If $G$ is connected and does not contain a mixed cycle, we say that $G$ is a {\it tree}. If $G$ is a tree, then a vertex $v \in V(G)$ with $d_G(v)=1$ is called a {\it leaf} of $G$.

\subsubsection{Non-standard graph theory and problem-specific definitions}
Given a mixed graph $G$ and a weight function $w:V(G)\rightarrow \mathbb{Z}_{\geq 0}$, we use $w(X)$ for $\sum_{x \in X}w(x)$ for some $X \subseteq V(G)$, we use $w(H)$ for $w(V(H))$ for some mixed subgraph $H$ of $G$, and we use $w(\mathcal{H})$ for  $\sum_{H \in \mathcal{H}}w(H)$ for some collection $\mathcal{H}$ of vertex-disjoint mixed subgraphs of $G$.

Given a mixed graph $G$, we use $P(G)$ for the set of ordered pairs of distinct vertices in $V(G)$. Given a function $g:P(G)\rightarrow \mathbb{Z}_{\geq 0}$ and some $(u,v)\in P(G)$, we abbreviate $g((u,v))$ to $g(u,v)$. We define the function $\kappa_{G}:P(G)\rightarrow \{0,1\}$ by $\kappa_G(u,v)=1$ if $v$ is reachable from $u$ in $G$, and $\kappa_G(u,v)=0$, otherwise, for all $(u,v)\in P(G)$. Observe that $\kappa_G(v,v)=1$ for all $v \in V(G)$.  We use $R(G)$ for $\sum_{(u,v)\in P(G)}\kappa_G(u,v)$. Further, given a weight function $w:V(G)\rightarrow \mathbb{Z}_{\geq 0}$, we use $R(G,w)$ for $2\sum_{v \in V(G)}{w(v)\choose 2}+\sum_{(u,v)\in P(G)}w(u)w(v)\kappa_G(u,v)$. Given a mixed graph $G$ and a constant $\alpha>0$, we say that an orientation $\vec{G}_0$ of $G$ is {\it $\alpha$-optimal} if $R(\vec{G}_0)\geq \alpha R(\vec{G})$ holds for every orientation $\vec{G}$ of $G$. Next, given a weight function $w:V(G)\rightarrow \mathbb{Z}_{\geq 0}$, we say that an orientation $\vec{G}_0$ of $G$ is {\it $\alpha$-optimal for $(G,w)$} if $R(\vec{G}_0,w)\geq \alpha R(\vec{G},w)$ holds for every orientation $\vec{G}$ of $G$. In either case, we use {\it optimal} for 1-optimal.

For some $u \in V(G)$, we use $Out_G(u)$ for all the vertices in $V(G)$ which are reachable from $u$ in $G$ and $In_G(u)$ for all the vertices in $V(G)$ from which $u$ is reachable. Observe that $u \in In_G(u)\cap Out_G(u)$.
Given two mixed graphs $G,H$, an orientation $\vec{G}$ of $G$, an orientation $\vec{H}$ of $H$ and a set $F \subseteq E(G)\cap E(H)$, we say that $\vec{G}$ and $\vec{H}$ {\it agree} on $F$ if all edges in $F$ have the same orientation in $\vec{G}$ and $\vec{H}$. Given two mixed graphs $T$ and $U$, we say that $U$ is a {\it mixed supergraph} of $T$ if $V(T)\subseteq V(U)$ and $U[V(T)]$ is a partial orientation of $T$. Observe that this does not necessarily mean that $T$ is a mixed subgraph of $U$.
Moreover, for a nonnegative integer $\ell$, we say that $U$ is a {\it mixed $\ell$-supergraph} of $T$ if $U$ is a mixed supergraph of $T$ and $|V(U)\setminus V(T)|\leq \ell$.

Given two digraphs $D_1$ and $D_2$, we use $D_1 \langle D_2 \rangle$ for the unique orientation of $UG(D_1)$ in which an edge $uv$ has its orientation  in $D_2$ if $\{u,v\}\subseteq V(D_2)$ and has its orientation in $D_1$, otherwise.

An instance of WAMMRO consists of a mixed graph $G$ and a weight function $w:V(G)\rightarrow \mathbb{Z}_{\geq 0}$. We say that $(G,w)$ is {\it connected} if $G$ is connected. An {\it undirected component} of $(G,w)$ is a component of $G\setminus A(G)$. Given two instances $(G_1,w_1)$ and $(G_2,w_2)$ of WAMMRO, we say that $(G_1,w_1)$ {\it extends} $(G_2,w_2)$ if $V(G_2)\subseteq V(G_1)$ and $w_1|V(G_2)=w_2$.

\subsubsection{Algorithms and Complexity}
 Given a problem instance of input size $n$ with a parameter $k$ and a function $g:\mathbb{Z}_{\geq 0}^2\rightarrow \mathbb{Z}_{\geq 0}$, we say that an algorithm runs in time $f(k)g(n,k)$ if there exists a computable function $f$ such that the algorithm runs in time $f(k)g(n,k)$. We wish to emphasize that, here, the role of $f$ and $g$ is not analogous: while $g$ is a fixed function, the usage of $f$ is short-hand for the existence of a computable function. For basics on the Landau notation and the definition of the complexity classes P and NP, see \cite{10.1016/j.dam.2011.09.003}. Given a problem instance of input size $n$ that comes with a parameter $k$, we say that the problem is XP (FPT) with respect to $k$ if there exists an algorithm that solves the problem and runs in time $n^{f(k)}$ $(f(k)n^{O(1)})$. For a maximization problem, an $\alpha$-approximation algorithm for some $\alpha<1$ is an algorithm that outputs a feasible solution whose objective value is at least $\alpha$ times as large as the value of any feasible solution. A maximization problem is called APX-hard if there exists some $\alpha<1$ such that the problem does not admit a polynomial-time $\alpha$-approximation algorithm unless P=NP.
\subsection{A satisfiability problem}\label{tdfzhijo}
Satisfiability problems are basic algorithmic problems whose role in computational complexity is immense. There are many variations of satisfiability problems. We here introduce one of them which will be used in the proof of \Cref{main1}. Namely, an instance of 3-Bounded Max 2-SAT (3BMax2Sat) consists of a set $X$ of binary variables and a set $\mathcal{C}$ of clauses each of which contains exactly two literals over $X$ such that for every $x \in X$, there exist at most 3 clauses in $\mathcal{C}$ that contain a literal of $x$. We say that an assignment $\phi:X \rightarrow \{True, False\}$ satisfies some $C \in \mathcal{C}$ if there exists some $x \in X$ such that either $x \in C$ and $\phi(x)=True$ or $\overline{x}\in C$ and $\phi(x)=False$. The reduction in the proof of \Cref{main1} is based on the fact that 3BMax2Sat is APX-hard. More precisely, we use the following result which is proven in \cite{bk}.

\begin{proposition}\label{max2}
Unless $P=NP$, there is no polynomial-time algorithm whose input is an instance $(X,\mathcal{C})$ of 3BMax2Sat and a positive integer $K$, that outputs 'yes' if $X$ admits an assignment satisfying at least  $K$ clauses of $\mathcal{C}$ and that outputs 'no' if every assignment of $X$ satisfies at most $\frac{2012}{2013}K$ clauses of $\mathcal{C}$.
\end{proposition}

We further need the following well-known result which can be obtained by basic probabilistic methods.
\begin{proposition}\label{34}
    For every instance $(X,\mathcal{C})$ of 3BMax2Sat, there exists an assignment of $X$ satisfying at least $\frac{3}{4}|\mathcal{C}|$ clauses of $\mathcal{C}$.
\end{proposition}

\subsection{An adapted version of MMRO}\label{equi}
It is mentioned in \cite{10.1007/978-3-540-87361-7_19} that for FPUMRO, one can assume that the input graph is acyclic as cycles can be safely contracted. Intuitively, the same should be possible for MMRO. However, due to the specific character of the problem, we have to associate a polynomial weight function to the newly created mixed graph. This weight function is associated to the mixed graph obtained from the contraction and displays how many vertices of the original mixed graph were contained in the strongly connected component which was contracted into that vertex. Formally, an instance of Weighted Acyclic Mixed Maximum Reachability Orientation (WAMMRO) consists of a mixed graph $G$ and a weight function $w:V(G)\rightarrow \mathbb{Z}_{\geq 0}$. The objective is to find an orientation $\vec{G}$ of $G$ that maximizes $R(\vec{G},w)$. Recall that $R(\vec{G},w)=2\sum_{v \in V(G)}{w(v)\choose 2}+\sum_{(u,v)\in P(G)}w(u)w(v)\kappa_{\vec{G}}(u,v)$. The first term in this expression reflects the contributions to the objective function made by vertices in the same strongly connected component of the original mixed graph. While this term is clearly the same for all orientations, it may not be neglected when considering approximation algorithms. Most results in Sections \ref{sec3} and \ref{sec4} will be proved for WAMMRO, which then implies the corresponding results for MMRO. The remainder of \Cref{equi} is dedicated to proving the equivalence of these problems. Here, we need to specify that when speaking about an instance $(G,w)$ of WAMMRO, we suppose that $w$ is given in unitary encoding. Further, we generally use $n$ for $|V(G)|+w(G)$. When we say that an algorithm runs in polynomial  time, we refer to $n$. As a side remark, we wish to point out that this specification is crucial. Indeed, as pointed out in \cite{Hakimi1997OrientingGT}, the Subset Sum problem can easily be reduced to the one of finding an orientation $\vec{G}$ maximizing $R(\vec{G},w)$ when given an undirected graph $G$ and a function $w:V(G)\rightarrow \mathbb{Z}_{\geq 0}$ given in binary encoding.

The first result justifies why the operation of contracting a mixed cycle behaves well with respect to the final solutions. In particular, we show that contracting a strongly connected directed subgraph of a given digraph does not affect the objective value when accurately adapting the weight function.
\begin{proposition}\label{contract}
    Let $D$ be a digraph, $w:V(D)\rightarrow \mathbb{Z}_{\geq 0}$ a weight function, $X \subseteq V(D)$, let a digraph $D'$ be obtained from $D$ by contracting $X$ into a new vertex $v_X$ and let $w':V(D')\rightarrow \mathbb{Z}_{\geq 0}$ be defined by $w'(v_X)=w(X)$ and $w'(v)=w(v)$ for all $v \in V(D')\setminus \{v_X\}$. Then $R(D',w')\geq R(D,w)$. Moreover, if $D[X]$ is strongly connected, then equality holds.
\end{proposition}
\begin{proof}
    Let $P_1$ contain the pairs $(u,v)\in P(D)$ with $\{u,v\}\subseteq X$ and let $P_2$ contain the pairs $ (u,v) \in P(D)$ with $\{u,v\}\cap X=\emptyset$. 

  First observe that
    \begin{equation}\label{firstp1}
    \begin{split}
       \sum_{(u,v)\in P_1}w(u)w(v)\kappa_{D}(u,v)+2\sum_{v \in X}{w(v)\choose 2}&\leq \sum_{(u,v)\in P_1}w(u)w(v)+2\sum_{v \in X}{w(v)\choose 2}\\
       &=\sum_{u \in X}w(u)(\sum_{v \in X\setminus \{u\}}w(v)+w(u)-1)\\
        &=w'(v_X)(w'(v_X)-1)\\
        &=2{w'(v_X)\choose 2}.
        \end{split}
    \end{equation}
    Moreover, if $D[X]$ is strongly connected, then $\kappa_{D}(u,v)=1$ for all $(u,v)\in P_1$, so equality holds.

Next consider some $(u,v)\in P_2$. If there exists a directed path from $u$ to $v$ in $D$, then there also exists a directed path from $u$ to $v$ in $D'$. Moreover, if $D[X]$ is strongly connected and there exists a directed path from $u$ to $v$ in $D'$, then there also exists a directed path from $u$ to $v$ in $D$.  This yields \begin{align}\label{firstp2}
    \sum_{(u,v)\in P_2}w(u)w(v)\kappa_{D}(u,v)\leq \sum_{(u,v)\in P_2}w'(u)w'(v)\kappa_{D'}(u,v)
\end{align}
and if $D[X]$ is strongly connected, then equality holds.

Further observe that for every $u \in V(D)\setminus X$, there exists a directed path from $u$ to $v$ in $D$ for some $v \in X$ if and only if there exists a directed path from $u$ to $v_X$ in $D'$. Moreover, if $D[X]$ is strongly connected and there exists a directed path from $u$ to $v$ in $D$ for some $v \in X$, then there exists a directed path from $u$ to $v'$ in $D$ for all $v' \in X$. 

We obtain that
\begin{equation}\label{firstp3}
\begin{split}
    \sum_{u \in V(D)\setminus X}\sum_{v \in X}w(u)w(v)\kappa_{D}(u,v)
    &\leq \sum_{u \in V(D)\setminus X}w(u)\sum_{v \in X}w(v)\max_{v' \in X}\kappa_{D}(u,v')\\
    &=\sum_{u \in V(D)\setminus X}w(u)\sum_{v \in X}w(v)\kappa_{D'}(u,v_C)\\
     &=\sum_{u \in V(D)\setminus X}w'(u)w'(v_C)\kappa_{D'}(u,v_C).
     \end{split}
\end{equation}
and if $D[X]$ is strongly connected, then equality holds throughout.
A similar argument shows that 

   \begin{align}\label{firstp4}
    \sum_{u \in X}\sum_{v \in V(D)\setminus X}w(u)w(v)\kappa_{D}(u,v)
     &\leq \sum_{v \in V(D)\setminus X}w'(v_C)w'(v)\kappa_{D'}(v_C,v).
\end{align}
and that, if $D[X]$ is strongly connected, then equality holds.

Summing \eqref{firstp1}, \eqref{firstp2},\eqref{firstp3}, and \eqref{firstp4}, we obtain that
\begin{align*}
    R(D,w)=&2\sum_{v \in V(D)}{w(v)\choose 2}+\sum_{(u,v)\in P(D)}w(u)w(v)\kappa_D(u,v)\\
    =&2\sum_{v \in X}{w(v)\choose 2}+2\sum_{v \in V(D)\setminus X}{w(v)\choose 2}+\sum_{(u,v)\in P_1}w(u)w(v)\kappa_D(u,v)\\&+\sum_{(u,v)\in P_2}w(u)w(v)\kappa_D(u,v)+ \sum_{u \in V(D)\setminus X}\sum_{v \in X}w(u)w(v)\kappa_{D}(u,v)\\&+ \sum_{u \in  X}\sum_{v \in V(D)\setminus X}w(u)w(v)\kappa_{D}(u,v)\\&\\
    \leq& 2{w'(v_C)\choose 2}+2\sum_{v \in V(D)\setminus X}{w'(v)\choose 2}+\sum_{(u,v)\in P_2}w'(u)w'(v)\kappa_{D'}(u,v)\\&+\sum_{u \in V(D)\setminus X}w'(u)w'(v_C)\kappa_{D'}(u,v_C)+\sum_{v \in V(D)\setminus X}w'(v_C)w'(v)\kappa_{D'}(v_C,v)\\
    =&2\sum_{v \in V(D)}{w'(v)\choose 2}+\sum_{(u,v)\in P(D')}w'(u)w'(v)\kappa_{D'}(u,v)\\
    =&R(D',w').
\end{align*}
Moreover, if $D[X]$ is strongly connected, then equality holds.
\end{proof}

In the remaining results, we aim to prove the equivalence of WAMMRO and MMRO. The first result shows that an instance of MMRO can be reduced to an equivalent instance of WAMMRO by contracting mixed cycles. It will be used several times in Section \ref{sec4}.
\begin{proposition}\label{addwa}
    Let $G$ be an instance of MMRO and let $n=|V(G)|$. Then, in polynomial time, we can compute an instance $(G_0,w)$ of WAMMRO with $|V(G_0)|\leq |V(G)|,|A(G_0)|\leq |A(G)|$, and $w(G_0)=n$ such that for every positive integer $K$, there exists an orientation $\vec{G}_0$ of $G_0$ with $R(\vec{G}_0,w)\geq K$ if and only if there exists an orientation $\vec{G}$ of $G$ with $R(\vec{G})\geq K$. Moreover, given an orientation $\vec{G}_0$ of $G_0$, we can compute an orientation $\vec{G}$ of $G$ with $R(\vec{G})\geq R(\vec{G}_0,w)$ in polynomial time.
\end{proposition}
\begin{proof}
    Let $G$ be an instance of MMRO. We now define a sequence $G^0,\ldots,G^q$ of mixed graphs and a sequence of weight functions $w^0:V(G^0)\rightarrow \mathbb{Z}_{\geq 0},\ldots,w^q:V(G^q)\rightarrow \mathbb{Z}_{\geq 0}$. We first set $G^0=G$ and define $w^0:V(G^0)\rightarrow \mathbb{Z}_{\geq 0}$ by $w^0(v)=1$ for all $v \in V(G^0)$. Now suppose that $G^0,\ldots,G^{i}$ and $w^0,\ldots,w^{i}$ have already been constructed for some $i\geq 0$. If  $G^{i}$ is acyclic, we do not continue the sequences. Otherwise, we find a mixed cycle $C$ in $G$. We now obtain $G^{i+1}$ from $G^{i}$ by contracting $C$ into a new vertex $v_C$. We further define $w^{i+1}:V(G^{i+1})\rightarrow {\mathbb{Z}_{\geq 0}}$ by $w^{i+1}(v_C)=w^{i}(C)$ and $w^{i+1}(v)=w^{i}(v)$ for all $v \in V(G^{i+1})\setminus \{v_C\}$. Observe that for every $i\geq 1$ for which $(G^{i},w^{i})$ is defined, we have $|V(G^i)|<|V(G^{i-1})|$. It hence follows that this procedure terminates and $q \leq n$. We set $(G_0,w)=(G^q,w^q)$. Observe that for $i \in [q]$, we clearly have that $(G^{i},w^{i})$ can be computed from $(G^{i-1},w^{i-1})$ in polynomial time. As $q \leq n$, we obtain that $(G_0,w)$ can be computed in polynomial time from $G$. Next, we recursively have $|V(G^i)|\leq |V(G^{i-1})|,|A(G^i)|\leq |A(G^{i-1})|$, and $w(G^i)= n$ for $i \in [q]$. This yields that $|V(G_0)|\leq |V(G)|,|A(G_0)|\leq |A(G)|$, and $w(G_0)=n$.

    Now consider some positive integer $K$ and suppose that there exists an orientation $\vec{G}$ of $G$ with $R(\vec{G})\geq K$. It follows directly by construction that there exists an orientation $\vec{G}^0$ of $G^0$ with $R(\vec{G}^0,w^0)\geq K$. We will inductively show that there exists an orientation $\vec{G}^{i}$ of $G^{i}$ with $R(\vec{G}^{i},w^{i})\geq K$. Consider some $i \in [q]$, let $C$ be the mixed cycle in $G^{i-1}$ used when constructing $G^{i}$. Inductively, we may suppose that there exists an orientation $\vec{G}^{i-1}$ of $G^{i-1}$ with $R(\vec{G}^{i-1},w^{i-1})\geq K$. Let $\vec{G}^{i}$ be the orientation of $G^{i}$ in which every edge of $G^{i}$ has the same orientation as in $\vec{G}^{i-1}$. It follows from Proposition \ref{contract} that $R(\vec{G}^{i},w^{i})\geq R(\vec{G}^{i-1},w^{i-1})\geq K$. Inductively, we obtain in particular that there exists an orientation $\vec{G}_0$ of $G_0$ with $R(\vec{G}_0,w_0)\geq K$.

    Now suppose that there exists an orientation $\vec{G}_0$ of $G_0$ with $R(\vec{G}_0,w_0)\geq K$. We now define a sequence $\vec{G}^q,\ldots,\vec{G}^0$ of digraphs such that $\vec{G}^{i}$ is an orientation of $G^{i}$ for $i \in [q]_0$. First let $\vec{G}^q=\vec{G}_0$. Now suppose that $\vec{G}^q,\ldots,\vec{G}^{i}$ has already been defined for some $i \in [q]$ and let $C$ be the mixed cycle that is used when constructing $G^{i}$ from $G^{i-1}$. As $C$ is a mixed cycle, we clearly have that a strongly connected orientation $\vec{C}$ of $C$ can be computed in polynomial time. We now let $\vec{G}^{i-1}$ be the unique orientation of $G^{i-1}$ in which all edges in $E(C)$ have the orientation they have in $\vec{C}$ and all edges in $E(G^{i-1})\setminus E(C)$ have the orientation they have in $\vec{G}^{i}$. As $\vec{C}$ is strongly connected, it follows from Proposition \ref{contract} that $R(\vec{G}^{i-1},w^{i-1})= R(\vec{G}^{i},w^{i})$. Moreover, it is not difficult to see that $\vec{G}^{i-1}$ can be computed from $\vec{G}^{i}$ in polynomial time. This finishes the description of $\vec{G}^q,\ldots,\vec{G}^0$. Finally, we set $\vec{G}=\vec{G}^0$. Observe that $\vec{G}$ is an orientation of $G$. Inductively, we obtain that $R(\vec{G})=R(\vec{G}^0,w^0)=R(\vec{G}_0,w_0)\geq K$. Moreover, as $q \leq n$, we obtain that the time for the whole procedure of computing $\vec{G}$ is polynomial.
\end{proof}

We now show the other direction, that is, that an instance of WAMMRO can be reduced to an instance of MMRO. This result will be applied in the hardness proof in Section \ref{sec3}. While it seems evident that a vertex of positive weight can be replaced by a collection of vertices of appropriate size, which are sufficiently connected among each other, the main difficulty is to incorporate vertices of weight 0. This additional difficulty leads to the following statement and its proof becoming a bit technical. For the next result, recall that, given an instance $(G,w)$ of WAMMRO, we use $n$ for $|V(G)|+w(G)$.
\begin{proposition}\label{smallbig}
    Let $(G,w)$ be an instance of WAMMRO. Then, in polynomial time, we can compute an instance $G_0$ of MMRO and an integer $K_0 \leq n^9$ such that for every positive integer $K$, there exists an orientation $\vec{G}_0$ of $G_0$ with $R(\vec{G}_0)\geq Kn^8+K_0$ if and only if there exists an orientation $\vec{G}$ of $G$ with $R(\vec{G},w)\geq K$. 
\end{proposition}
\begin{proof}
Let $V_0=w^{-1}(0)$ and $V_1 =V(G)\setminus V_0$. We further let $P_1$ contain the pairs $(u,v)\in P(G)$ with $\{u,v\}\subseteq V_1$ and we set $P_0=P(G)\setminus P_1$. We further set $K_0=2\sum_{v \in V_1}{n^4w(v)\choose 2}-n^8{w(v)\choose 2}$. Observe that 
\begin{align*}
    K_0&=2\sum_{v \in V_1}{n^4w(v)\choose 2}-n^8{w(v)\choose 2}\\
    &=2\sum_{v \in V_1}
\frac{1}{2}w(v)n^4(n^4-1)\\
&\leq \sum_{v \in V_1}w(v)n^8\\
&\leq n^9.\end{align*}
We now define $w':V(G)\rightarrow\mathbb{Z}_{\geq 0}$ by $w'(v)=1$ for all $v \in V_0$ and $w'(v)=w(v)n^4$ for all $v \in V_1$.
    We obtain $G_0$ from $G$ in the following way: for every $v \in V(G)$, we add a set $X_v$ of $w'(v)-1$ vertices and for every $v' \in X_v$, we add a digon linking $v$ and $v'$. For simplicity, we use $X'_v$ for $X_v \cup \{v\}$. This finishes the description of $G_0$. Observe that $G_0$ and $K_0$ can be computed in polynomial time from $G$. Let $w_0:V(G_0)\rightarrow \mathbb{Z}_{\geq 0}$ be the function defined by $w_0(v)=1$ for all $v \in V(G_0)$. 

Next observe that there exists a direct bijection between the orientations of $G$ and the orientations of $G_0$. We say that an orientation $\vec{G}$ of $G$ and an orientation $\vec{G}_0$ of $G_0$ are {\it associated} if they agree on $E(G)$.

Now let $\vec{G}$ be an orientation of $G$ and $\vec{G}_0$ the associated orientation of $G_0$. By a repeated application of Proposition \ref{contract} and the definitions of $R$ and $K_0$, we have 
\begin{align*}
    R(\vec{G}_0)&=R(\vec{G}_0,w_0)\\
    &=R(\vec{G},w')\\
    &=2\sum_{v \in V(G)}{w'(v)\choose 2}+\sum_{(u,v) \in P(G)}w'(u)w'(v)\kappa_{\vec{G}}(u,v)\\
       &=2\sum_{v \in V(G)}{n^4 w(v)\choose 2}+\sum_{(u,v) \in P_1}w'(u)w'(v)\kappa_{\vec{G}}(u,v)+\sum_{(u,v) \in P_0}w'(u)w'(v)\kappa_{\vec{G}}(u,v)\\
       &=K_0+2n^8\sum_{v \in V(G)}{w(v)\choose 2}+n^8\sum_{(u,v) \in P_1}w(u)w(v)\kappa_{\vec{G}}(u,v)+\sum_{(u,v) \in P_0}w'(u)w'(v)\kappa_{\vec{G}}(u,v)\\
       &=n^8R(\vec{G},w)+K_0+\sum_{(u,v) \in P_0}w'(u)w'(v)\kappa_{\vec{G}}(u,v).
\end{align*}

    Further observe that $\sum_{(u,v) \in P_0}w'(u)w'(v)\kappa_{\vec{G}}(u,v)\leq n^5|P(G)|<n^8$. Hence, as $R(\vec{G},w)$ is an integer, the statement follows.
\end{proof}
\section{Hardness result}\label{sec3}
This section is dedicated proving \Cref{main1}. The proof is cut in two parts. First, in \Cref{apxwa}, we prove the APX-hardness of WAMMRO. This part is a reduction from 3BMax2Sat and the more challenging part. Our reduction is similar to the one proving the NP-hardness of the complete solution version of FPMMRO in \cite{ARKIN2002271}. However, some more care is needed. Roughly speaking, in both reductions, we actually want to only consider the reachabilities among certain pairs corresponding to clauses being satisfied. While in the reduction in \cite{ARKIN2002271} the restriction to these pairs can be obtained from the problem definition, for the hardness of WAMMRO, we need to make sure by some additional connections that we can control the contribution of the remaining pairs to the objective function. Further, the restricted number of variable occurences allows us to make sure that these additional reachabilities do not dominate those obtained by the pairs of vertices we actually want to consider, which is crucial for the APX-hardness. 

In the second part of the proof, we combine  \Cref{apxwa} and \Cref{smallbig} to obtain \Cref{main1}. While this second part is calculation heavy, it does not present any conceptual difficulties.

We now give the hardness proof for WAMMRO.
\begin{lemma}\label{apxwa}
Unless $P=NP$, there is no polynomial-time algorithm whose input is an instance $(G,w)$ of WAMMRO  and a positive integer $K\geq \frac{1}{16}n$, that outputs 'yes' if $G$ admits an orientation $\vec{G}$ with $R(\vec{G},w)\geq K$, and that outputs 'no' if $R(\vec{G},w)<\frac{34220}{34221}K$ holds for every orientation $\vec{G}$ of $G$.
\end{lemma}
\begin{proof}
    Suppose that an algorithm $A$ with the properties described in \Cref{apxwa} exists. In the following, we describe an algorithm based on $A$ whose input is an instance of 3BMax2Sat and a positive integer $K$. We will show that the algorithm has a behavior such that $P=NP$ follows by \Cref{max2}. 

    Let $(X,\mathcal{C})$ be an instance of 3BMax2Sat and $K$ a positive integer. Clearly, we may suppose that every variable of $X$ is contained in at least one clause. As every clause contains exactly two variables, we have $|\mathcal{C}|\geq \frac{1}{2}|X|$. Further, as every variable of $X$ is contained in at most 3 clauses of $\mathcal{C}$, we have $|\mathcal{C}|\leq 3|X|$. If $K \leq \frac{3}{4}|\mathcal{C}|$, we let the algorithm output 'yes'. This is an output of the desired form by \Cref{34}. We may hence suppose that $K > \frac{3}{4}|\mathcal{C}|$. 

In the following, we let $\mathcal{Y}$ denote the set of subsets of $\mathcal{C}$ that consist of exactly two distinct clauses $C,C' \in \mathcal{C}$ that share at least one variable. As every variable of $X$ is contained in at most 3 clauses of $\mathcal{C}$, we have $|\mathcal{Y}|\leq 3|X|$.
    We now construct an instance $(G,w)$ of WAMMRO. First, we let $V(G)$ contain a set $V_\mathcal{C}$ that contains two vertices $a_C$ and $b_C$ for every $C \in \mathcal{C}$. For every $\{C,C'\}\in \mathcal{Y}$, we let $A(G)$ contain the arcs $a_Cb_{C'}$ and $a_{C'}b_C$. Next, we let $V(G)$ contain a set $V_X$ that contains two vertices $s_x$ and $s_{\overline{x}}$ for every $x\in X$. For every $x \in X$, we let $E(G)$ contain an edge linking $s_x$ and $s_{\overline{x}}$. Further, for every $x \in X$ and every $C \in \mathcal{C}$ such that $x \in C$, we let $A(G)$ contain the arcs $a_Cs_{\overline{x}}$ and $s_xb_C$ and for every $x \in X$ and every $C \in \mathcal{C}$ such that $\overline{x} \in C$, we let $A(G)$ contain the arcs $a_Cs_x$ and $s_{\overline{x}}b_C$. Finally, we set $w(v)=1$ for all $v \in V_\mathcal{C}$ and $w(v)=0$ for all $v \in V_X$. This finishes the description of $(G,w)$. We use $n$ for $|V(G)|+w(G)$. For an illustration, see Figure \ref{drftg}.
    
\begin{figure}[h]
    \centering
        \includegraphics[width=.8\textwidth]{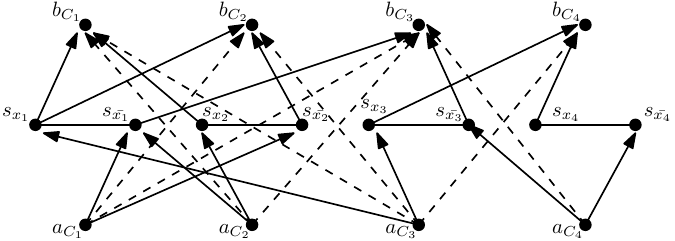}
        \caption{An example of the construction in the proof of \Cref{apxwa}, where $X=\{x_1,x_2,x_3,x_4\}$ and $\mathcal{C}=\{C_1=\{x_1,x_2\},C_2=\{x_1,\overline{x_2}\}, C_3=\{\overline{x_1},\overline{x_3}\},C_4=\{x_3,x_4\}\}$. For the sake of better readability, the arcs with both endvertices in $V_{\mathcal{C}}$ are dashed while the arcs with one endvertex in $V_{\mathcal{C}}$ and one endvertex in $V_X$ are solid. }\label{drftg}
\end{figure}

    Observe that $a_C$ is a source in $G$ for all $C \in \mathcal{C}$, $b_C$ is a sink in $G$ for all $C \in \mathcal{C}$ and $G[V_x]$ is a matching. It follows that $G$ is acyclic, so $(G,w)$ is indeed an instance of WAMMRO.  Further, it is not difficult to see that $(G,w)$ can be constructed in polynomial time from $(X,\mathcal{C})$. Further, we set $K_0=2|\mathcal{Y}|+K$. As $K >\frac{3}{4}|\mathcal{C}|, |\mathcal{C}|\geq \frac{1}{2}|X|$, by construction, and $w(v)\leq 1$ for all $v \in V(G)$, we have that $K_0\geq K\geq \frac{3}{4}|\mathcal{C}|\geq \frac{1}{4}(|X|+|\mathcal{C})|= \frac{1}{8}|V(G)|\geq \frac{1}{16}n$. Further, as  $|\mathcal{Y}|\leq 3|X|\leq 6 |\mathcal{C}|\leq 8K$, we have $K_0=2|\mathcal{Y}|+K\leq 17K$.

    We now apply $A$ to $(G,w)$ and $K_0$ and output the output of $A$. This finishes the description of our algorithm. We still need to show that it has the desired properties. First observe that, as $(G,w)$ can be constructed in polynomial time from $(X,\mathcal{C})$ and $A$ runs in polynomial time, we obtain that the entire algorithm runs in polynomial time.

    Next suppose that there exists an assignment $\phi:X \rightarrow \{True,False\}$ that satisfies a set $\mathcal{C}_0$ of at least $K$ clauses of $\mathcal{C}$.

    We now define an orientation $\vec{G}$ of $G$ in the following way: for every $x \in X$, we orient the edge linking $s_x$ and $s_{\overline{x}}$ from $s_{\overline{x}}$ to $s_x$ if $\phi(x)=True$ and from $s_{x}$ to $s_{\overline{x}}$ if $\phi(x)=False.$

    We show in the following that $R(\vec{G},w)\geq K_0$. First observe that for every $\{C,C'\}\in \mathcal{Y}$, we have that $\kappa_{\vec{G}}(a_{C'},b_C)=\kappa_{\vec{G}}(a_C,b_{C'})=1$. Next, let $C \in \mathcal{C}_0$. As $C$ is satisfied by $\phi$, there exists some $x \in X$ such that either $x \in C$ and $\phi(x)=True$ or $\overline{x}\in C$ and $\phi(x)=False$. In the first case, by construction, we obtain that $a_Cs_{\overline{x}}s_xb_C$ is a directed $a_Cb_C$-path in $\vec{G}$ and in the second case, again by construction, we obtain that $a_Cs_xs_{\overline{x}}b_C$ is a directed $a_Cb_C$-path in $\vec{G}$. In either case, we have that $\kappa_{\vec{G}}(a_C,b_{C})=1$.  
    It follows that \begin{align*}
        R(\vec{G},w)=&2\sum_{v \in V(G)}{w(v)\choose 2}+\sum_{(u,v)\in P(G)}\kappa_{\vec{G}}(u,v)w(u)w(v)\\
        \geq& \sum_{\{C,C'\}\in \mathcal{Y}}\left(\kappa_{\vec{G}}(a_{C'},b_C)+\kappa_{\vec{G}}(a_{C},b_{C'})\right)+\sum_{C \in \mathcal{C}_0}\kappa_{\vec{G}}(a_{C},b_{C})\\
        \geq& 2|\mathcal{Y}|+|\mathcal{C}_0|\\
        \geq &2|\mathcal{Y}|+K\\
        =&K_0.
    \end{align*}

    It follows by assumption that our algorithm outputs 'yes'.

    Now suppose that our algorithm does not output 'no'. It follows by construction that $A$ does not output 'no' when applied to $(G,w)$. We obtain by the assumption on $A$ that there exists an orientation $\vec{G}$ of $G$ such that $R(\vec{G},w)\geq \frac{34220}{34221}K_0$. We now define an assignment $\phi:X \rightarrow \{True, False\}$ in the following way: for every $x \in X$ such that the edge linking $s_x$ and $s_{\overline{x}}$ is oriented from $s_{\overline{x}}$ to $s_x$, we set $\phi(x)=True$, and for every $x \in X$ such that the edge linking $s_x$ and $s_{\overline{x}}$ is oriented from $s_x$ to $s_{\overline{x}}$, we set $\phi(x)=False$. Let $\mathcal{C}_0$ be the set of clauses satisfied by $\phi$. Suppose for the sake of a contradiction that $|\mathcal{C}_0|<\frac{2012}{2013} K$.

    Let $P_1$ be the set of pairs of the form $(a_C,b_{C'})$ with $\{C,C'\}\in \mathcal{Y}$ and let $P_2$ contain the pairs of the form $(a_C,b_C)$ for $C \in \mathcal{C}_0$. Observe that $|P_1|=2|\mathcal{Y}|$.  Observe that, by construction, we have $\kappa_{\vec{G}}(u,v)=0$ for all $(u,v)\in P(G)\setminus(P_1 \cup P_2)$ with $\min \{w(u),w(v)\}\neq 0$.

By the definition of $R$, the definition of $w$, construction, $|\mathcal{C}_0|<\frac{2012}{2013}K$, the definition of $K_0$, and $K \geq \frac{1}{17}K_0$, we obtain that 
    \begin{align*}
        R(\vec{G},w)&=2\sum_{v \in V(G)}{w(v)\choose 2}+\sum_{(u,v)\in P(G)}\kappa_{\vec{G}}(u,v)w(u)w(v)\\
        &=\sum_{(u,v)\in P_1}w(u)w(v)+\sum_{(u,v)\in P_2}w(u)w(v)\\
        &=|P_1|+|P_2|\\
        &=2|\mathcal{Y}|+|\mathcal{C}_0|\\
        &<2|\mathcal{Y}|+\frac{2012}{2013}K\\
        &=K_0-\frac{1}{2013} K\\
        &\leq \frac{34220}{34221}K_0.
    \end{align*}

    We obtain a contradiction to the assumption on $\vec{G}$. This yields that $|\mathcal{C}_0|\geq \frac{2012}{2013}k$.

    It follows that our algorithm has the properties described in \Cref{max2}. We hence obtain that $P=NP$.
\end{proof}

We are now ready to give the proof of \Cref{main1}. For the sake of clarity, we prove the following more formal restatement.

\begin{theorem}\label{main1tech}
Unless $P=NP$, there is no polynomial-time algorithm whose input is an instance $G_0$ of MMRO  and a positive integer $K$, that outputs 'yes' if $G_0$ admits an orientation $\vec{G}_0$ with $R(\vec{G}_0)\geq K$, and that outputs 'no' if $R(\vec{G}_0)<\frac{581773}{581774}K$ for every orientation $\vec{G}_0$ of $G$.
\end{theorem}
\begin{proof}
Suppose that there exists an algorithm $A$ with the properties described in \Cref{main1tech}. We will show that P=NP using \Cref{apxwa}. Let $(G,w)$ be an instance of WAMMRO and $K_1\geq \frac{1}{16}n$ a positive integer. If $n\leq 18737776991$, we can solve the problem by a brute force approach. We may hence suppose that $n \geq 18737776992$. Now, by Proposition \ref{smallbig}, in polynomial time, we can compute an instance $G_0$ of MMRO and a constant $K_0$ with $K_0 \leq n^9$ such that for any positive integer $K$, there exists an orientation $\vec{G}$ of $G$ with $R(\vec{G},w)\geq K$ if and only if there exists an orientation $\vec{G}_0$ of $G_0$ with $R(\vec{G}_0)\leq =Kn^8+K_0$. Let $K'=K_1n^8+K_0$. We now apply $A$ to $G_0$ and $K'$. We output the output of $A$. This finishes the description of our algorithm. As $A$ runs in polynomial time and by construction, we obtain that the entire algorithm runs in polynomial time.

Next suppose that there exists an orientation $\vec{G}$ of $G$ with $R(\vec{G},w)\geq K_1$. Then, by construction, there exists an orientation $\vec{G}_0$ of $G_0$ with $R(\vec{G}_0)\geq K'$. It follows by assumption that $A$ outputs 'yes' and hence our algorithm outputs 'yes'.

Now suppose that our algorithm does not output 'no'. As $K_1n^8\geq \frac{1}{16}n^9\geq \frac{1}{16}K_0$, we have $\frac{1}{581774} K'\leq \frac{17}{581774} K_1n^8=\frac{1}{34222}K_1n^8$. We obtain by assumption that there exists an orientation $\vec{G}_0$ of $G_0$ with $R(\vec{G}_0)\geq \frac{581773}{581774}K'\geq \frac{34221}{34222} K_1n^8+K_0$. We hence obtain by assumption that there exists an orientation $\vec{G}$ of $G$ with $R(\vec{G},w)\geq \lfloor \frac{34221}{34222}K_1\rfloor$. As $K_1\geq \frac{1}{16}n\geq 1171111062$, we obtain that $R(\vec{G},w)\geq \lfloor \frac{34221}{34222}K_1\rfloor\geq \frac{34221}{34222}K_1-1\geq \frac{34220}{34221}K_1$. We hence obtain that $P=NP$ by \Cref{apxwa}.
\end{proof}
\section{Algorithmic results}\label{sec4}
The objective of this Section is to prove Theorems \ref{main2} and \ref{main3}.  We now give an overview of these proofs. We mainly explain the proof strategy for Theorem \ref{main2} and later point out how to adapt it for the algorithm of \Cref{main3}.

 Roughly speaking, the main strategy of the algorithm is to cut a given instance into small parts which are highly structured, then compute optimal solutions for these parts and put them together. One important substructure are so-called undirected components. Namely, given an instance $(G,w)$ of WAMMRO, an {\it undirected component} of $G$ is a component of $G\setminus A(G)$. 

 Consider an undirected component $T$ of $G$. Observe that $T$ is a tree as $G$ is acyclic. Our objective is to find a set of orientations $\mathcal{T}$ of $T$ such that for any orientation $\vec{G}$ of $G$, there exists some $\vec{T} \in \mathcal{T}$ with the property that when reorienting all edges in $E(T)$ according to $\vec{T}$, we  obtain an orientation that is at least as good as $\vec{G}$. Formally, given an instance $(G,w)$ of WAMMRO and an undirected component $T$ of $G$, we say that a set $\mathcal{T}$ of orientations of $T$ is an {\it optimal replacement set} for $((G,w),T)$ if for every orientation $\vec{G}$ of $G$, there exists some $\vec{T} \in \mathcal{T}$ such that $R(\vec{G}\langle\vec{T}\rangle,w)\geq R(\vec{G},w)$. Here, recall that $\vec{G}\langle\vec{T}\rangle$ is the unique orientation of $G$ in which all edges in $E(T)$ have the orientation they have in $\vec{T}$ and all other edges in $E(G)$ have the orientation they have in $\vec{G}$.

 The importance of optimal replacement sets can be summarized in the following result.
 \begin{restatable}
 {lemma}{utgzhu}\label{utgzhu}
     Let $(G,w)$ be an instance of WAMMRO, let $T_1,\ldots,T_q$ be the undirected components of $G$ and for $i \in [q]$, let $\mathcal{T}_i$ be an optimal replacement set for $((G,w),T_i)$. Then there exists an optimal orientation $\vec{G}$ for $(G,w)$ such that $\vec{G}[V(T_i)]\in \mathcal{T}_i$ for $i \in [q]$. 
\end{restatable}

 The proof of \Cref{utgzhu} is not difficult. Observe that \Cref{utgzhu} suggests an algorithm for WAMMRO: given an instance $(G,w)$ of WAMMRO, find an optimal replacement set of $T$ of appropriate size for all undirected components $T$ of $G$ and then try all possible orientations of $G$ such that the inherited orientations of the undirected componets come from the optimal replacement sets. It is not difficult to prove that this yields an optimal orientation for $(G,w)$ indeed.

 The difficulty now is to prove that optimal replacement sets of appropriate size exist for all undirected components of $G$ and can be computed sufficiently fast.  As this seems challenging for general instances of WAMMRO, we first restrict the structure of the instances in consideration by a preprocessing step.

Roughly speaking, we consider instances in which the interaction of every undirected component with the remaining part of the mixed graph is limited. Formally, we say that an instance $(G,w)$ of WAMMRO is {\it dismembred} if every undirected component of $G$ either contains at most one vertex of $V(A(G))$ or contains exactly two vertices of $V(A(G))$ and each of them is incident to exactly one arc of $A(G)$. We show that in order to prove \Cref{main2}, it suffices to prove a version for dismembered instances. Moreover, we make the simple observation that we can restrict ourselves to connected instances, where an instance $(G,w)$ is {\it connected} if $UG(G)$ is connected. In order to be able to reuse the result later for the proof of \Cref{main3}, the following strong version taking into account approximate solutions is given.

\begin{restatable}{lemma}{reddismem}\label{reddismem}
    Let $g:\mathbb{Z}_{\geq 0}\times \mathbb{Z}_{\geq 0}\times [0,1]\rightarrow \mathbb{Z}_{\geq 0}$ be a fixed function and suppose that there exists an algorithm $A$ that computes a $(1-\epsilon)$-optimal orientation for every connected, dismembered instance of WAMMRO and every $\epsilon\geq 0$ in time $g(n,k,\epsilon)$. Then, for every $\epsilon\geq 0$, a $(1-\epsilon)$-optimal orientation for every instance of WAMMRO can be computed in time $f(k)g(n,k',\epsilon)n^{O(1)}$ for some $k'=O(k)$.
\end{restatable}

It turns out that for dismembered instances, the situation is much brighter and we can indeed efficiently compute optimal replacement sets of polynomial size. The following is the main lemma for the proof of \Cref{main2}.

\begin{restatable}{lemma}{rdztghok}\label{rdztghok}
    Let $(G,w)$ be a connected, dismembered instance of WAMMRO and let $T$ be an undirected component of $G$. Then, in polynomial time, we can compute an  optimal replacement set $\mathcal{T}$ for $((G,w),T)$ of size $O(n^3)$. 
\end{restatable}

\Cref{rdztghok}, \Cref{utgzhu} and \Cref{reddismem} together with \Cref{addwa} will imply \Cref{main2}.
\bigskip

We now give an overview of the proof of \Cref{rdztghok}. Consider an undirected component $T$ of a dismembered, connected instance $(G,w)$ of WAMMRO. Our strategy is to compute a collection of instances of WAMMRO extending $(T,w|V(T))$ such that a collection of orientations of $T$ consisting of the one inherited from an optimal solutions to each of these instances forms an optimal replacement set for $((G,w),T)$. Informally, the idea is that, given an orientation of $E(G)\setminus E(T)$, we want to consider a small mixed supergraph $U$ of $T$ together with a weight function $w'$ on $V(U)$ such that finding an optimal way to extend the orientation of $E(G)\setminus E(T)$ to $T$ corresponds to finding an optimal orientation of $(U,w')$. Formally, a {\it simulation set} for $((G,w),T)$ is a set $\mathcal{U}$ of mixed supergraphs of $T$ such that for every orientation $\vec{G}$ of $G$, there exists some $U \in \mathcal{U}$ and a weight function $w':V(U)\rightarrow [n]_0$ with $|w'|\leq |w|$ that is consistent with $w$ and with the property that for every  orientation $\vec{U}$ of $U$, we have that $\vec{G}\langle\vec{U}\rangle$ is an orientation of $G$, $\vec{U}\langle\vec{G}\rangle$ is an orientation of $U$ and $R(\vec{G}\langle\vec{U}\rangle,w)- R(\vec{G},w)\geq R(\vec{U},w')-R(\vec{U}\langle \vec{G}\rangle,w' )$ holds. 

Once we have a good simulation set $\mathcal{U}$ for $((G,w),T)$, we want to compute an optimal replacement set $\mathcal{T}$ in the following way. For every $U \in \mathcal{U}$ and every appropriate weight function $w'$ on $V(U)$ consistent with $w$, we compute an optimal orientation for $(U,w')$ and add the inherited orientation $\vec{T}$ of $T$ to $\mathcal{T}$. It is not difficult to see that $\mathcal{T}$ is an optimal replacement set indeed.

However, in order to obtain an efficient algorithm from this observation, we need to take care of several things. First, we need to make sure that the number of mixed graphs in $\mathcal{U}$ is small enough. Next, we need to make sure that the mixed graphs in $\mathcal{U}$ contain only slightly more vertices than $T$ in order to be able to enumerate all weight functions extending $w$. Finally, we need to make sure that for every $U \in \mathcal{U}$ and for every weight function $w'$ on $V(U)$, we can compute an optimal orientation for $(U,w')$ sufficiently fast. To this end, we need some more technical definitions. First, for some positive integer $\ell$, an {\it $\ell$-simulation set} for $((G,w),T)$ is a simulation set for $((G,w),T)$ such that every $U \in \mathcal{U}$ is a mixed $\ell$-supergraph of $T$. We further need to define a class of tree-like mixed graphs. Namely, we say that a mixed graph $U$ is {\it arboresque} if $A(G)$ contains an arc $rs$ such that the underlying graph of $G\setminus \{rs\}$ is a tree and one of $d_G^+(r)=d_{UG(G)}(r)=2$ and $d_G^-(s)=d_{UG(G)}(s)=2$ holds. We say that a collection $\mathcal{U}$ of mixed graphs is {\it arboresque} if every element of $U$ is arboresesque.  

We are now ready to formulate our main result on simulation sets.
\begin{restatable}
    {lemma}{rdtfguh}\label{rdtfguh}
    Let $(G,w)$ be a dismembered instance of WAMMRO and $T$ an undirected component of $G$ with $|V(T)\cap V(A(G))|\geq 1$. Then, in polynomial time, we can compute an arboresque 3-simulation set of size at most 2 for $((G,w),T)$. 
\end{restatable}

In order to make use of \Cref{rdtfguh}, we need to be able to efficiently handle arboresque mixed graphs.
Arboresque mixed graphs are a rather minor extension of mixed graphs whose underlying graphs are trees. We use this fact to show that an optimal orientation for an instance $(U,w)$ of WAMMRO such that $U$ is arboresque can be computed in polynomial time, using a dynamic programming approach. \begin{restatable}  
{lemma}{tfzghik}\label{tfzghik}
    Let $(G,w)$ be an arboresque instance of WAMMRO.  Then  an optimal orientation $\vec{G}$ for $(G,w)$ can be computed in polynomial time.
\end{restatable}
 We wish to remark that \Cref{tfzghik} is already a strengthening of the main result of \cite{Hakimi1997OrientingGT}. \Cref{tfzghik} and \Cref{rdtfguh} imply \Cref{rdztghok}. This completes the overview of the proof \Cref{main2}.\medskip

We now give an overview of how to adapt these results to prove \Cref{main3}. We prove approximation analogues of \Cref{utgzhu} and \Cref{rdztghok}. The analogue of \Cref{utgzhu} is rather straight forward and can be proved in a similar way. For the analogue of \Cref{rdztghok}, instead of an optimal replacement set $\mathcal{T}$ for $((G,w),T)$ whose size is polynomial in $n$, we obtain a set $\mathcal{T}$ of orientations of $T$ such that replacing with an orientation from this set may only yield an orientation that is by a small additive constant worse than the first orientation in consideration. On the other hand, we have that $|\mathcal{T}|$ is bounded by $f(k,\epsilon)$. This result is obtained from Lemmas \ref{rdtfguh} and \ref{tfzghik} in a similar way as \Cref{rdztghok}.
 The only difference is that instead of trying all possible weight functions for the mixed graphs in the simulation sets, we restrict to a smaller set of functions, avoiding to run the algorithm several times for very similar weight functions. We then conclude that an orientation that is only a certain additive amount away from an optimal solution can be computed in the desired running time.

 However, we need to prove a small multiplicative loss rather than an additive one. To this end, we prove a result showing that for every connected, dismembered instance of WAMMRO, there exists an orientation $\vec{G}$ of $G$ such that $R(\vec{G},w)$ is above a certain minimum threshold. In order to so, we focus on a largest undirected component of $G$ and analyze the objective value of an optimal solution as described in \cite{Hakimi1997OrientingGT}. Together with the additive approximation result described above, \Cref{main3} follows readily.

The rest of \Cref{sec4} is structured as follows. First, in \Cref{sec41}, we prove \Cref{reddismem}. Next, in \Cref{sec43}, we prove \Cref{rdtfguh}. In Section \ref{sec42}, we give the algorithm proving \Cref{tfzghik}. After, in \Cref{secml}, we conclude our main lemmas, that is, we prove \Cref{rdztghok} and its approximate analogue we need for the proof of \Cref{main3}. In \Cref{sec44}, we prove \Cref{utgzhu} and conclude \Cref{main2}. Finally, in Section \ref{sec45}, we make the necessary adaptions to prove \Cref{main3}.
 
\subsection{Restriction to dismembered instances}\label{sec41}
This section is dedicated to justifying that it suffices to prove our main results for connected, dismembered instances of WAMMRO. 
More concretely, we prove \Cref{reddismem}.

The strategy for the proof of \Cref{reddismem} can roughly be described as follows. Given a connected instance $(G,w)$ of WAMMRO, we wish to find a sufficiently small collection of edges in $E(G)$ whose orientation we guess and that has the property that each of the thus obtained partial orientations of $G$ together with $w$ forms a dismembered instance of WAMMRO. Formally, given an undirected graph $G$ and some $X \subseteq V(G)$, we say that a set $F \subseteq E(G)$ is {\it dismembering} for $(G,X)$ if every component of $G\setminus F$ either contains at most one vertex of $X \cup V(F)$ or it contains no vertex of $X$ and exactly two vertices of $V(F)$ each of which is incident to exactly one edge in $F$. The following result is the main lemma for our proof.
\begin{restatable}{lemma}{edgedis}\label{edgedis}
    Let $G$ be an undirected graph and let $F_0 \subseteq E(G)$ be an edge set such that every component of $G\setminus F_0$ is a tree. Then there exists a dismembering set $F$ for $(G\setminus F_0,V(F_0))$ with $|F|\leq 10 |F_0|$. Moreover, we can compute $F$ in polynomial time.
\end{restatable}

With \Cref{edgedis} at hand, it is not difficult to prove \Cref{reddismem}. Given a connected instance $(G,w)$, we use \Cref{edgedis} to compute a dismembering set $F$ for $(UG(G),V(F_0))$, where $F_0$ is the set of edges in $E(UG(G))$ that correspond to arcs in $A(G)$. We then obtain a collection of dismembered instances of WAMMRO by trying all possible orientations of the edges in $F$ and solve all obtained instances by the algorithm whose existence we suppose. We output the best orientation of $G$ we find. As every orientation of one of the created mixed graphs is an orientation of $G$ and every orientation of $G$ is an orientation of one of the created mixed graphs, we find an optimal orientation for $(G,w)$ during this procedure. Finally, we show that the condition of $G$ being connected can also be omitted.
\medskip

In the following, we first aim to prove \Cref{edgedis}. We will prove a corresponding statement for trees and then conclude \Cref{edgedis}. We first need a simple observation on trees. Here, given a tree $T$, we use $E_1(T)$ for the set of edges in $E(T)$ which are incident to a leaf of $T$ and we use $E_{\geq 3}(T)$ for the set of edges in $E(T)$ which are incident to some $v \in V(T)$ with $d_T(v)\geq 3$.
\begin{proposition}\label{treeedge}
    Let $T$ be a tree. Then $|E_{\geq 3}(T)|\leq 3|E_{1}(T)|$.
\end{proposition}
\begin{proof}
    Suppose otherwise and let $T$ be a minimum counterexample. Let $v$ be a leaf of $T$. If $T$ is a path, there is nothing to prove. We may hence suppose that $T$ is not a path and so there exists a unique vertex $x \in V(T)$ with $d_T(x)\geq 3$ that is linked to $v$ by a path $P$ in $T$ such that $d_T(u)=2$ holds for all $u \in V(P)\setminus\{v,x\}$. Let $T'=T-(V(P)\setminus\{x\})$.
    Observe that $T'$ is a tree by construction and hence, as $T'$ is smaller than $T$, we have $|E_{\geq 3}(T')|\leq 3|E_{1}(T')|$. Next, by the definition of $v,x$, and $P$, we have that $|E_{1}(T)|=|E_{1}(T')|+1$. Next, if $d_T(x)\geq 4$, we obtain that the only edge in $E_{\geq 3}(T)\setminus E_{\geq 3}(T')$ is the unique edge in $E(P)$ incident to $x$.
    This yields $|E_{\geq 3}(T)|=|E_{\geq 3}(T')|+1$. If $d_T(x)=3$, then observe that every edge in $E_{\geq 3}(T)\setminus E_{\geq 3}(T')$ is incident to $x$. This yields $|E_{\geq 3}(T)|=|E_{\geq 3}(T')|+3$. In either case, we have $|E_{\geq 3}(T)|\leq |E_{\geq 3}(T')|+3$. We obtain $|E_{\geq 3}(T)|\leq |E_{\geq 3}(T')|+3\leq 3|E_{1}(T')|+3=3(|E_{1}(T)|-1)+3=3|E_{1}(T)|$, a contradiction to $T$ being a counterexample.
\end{proof}

We next conclude an upper bound on the size of dismembering sets in trees.
\begin{lemma}\label{vertexdis}
    Let $T$ be a tree and $X \subseteq V(T)$. Then there exists a dismembering set $F$ for $(T,X)$ with $|F|\leq 5|X|$. Moreover, such a set can be computed in polynomial time.
\end{lemma}
\begin{proof}
    Suppose otherwise and let $(T,X)$ be a counterexample minimizing $|V(T)|$.
    \begin{claim}\label{leafx}
        Every leaf of $T$ is contained in $X$.
    \end{claim}
    \begin{proof}
        Suppose otherwise, so there exists a leaf $v$ of $T$ that is contained in $V(T)\setminus X$. As $|V(T- v)|<|V(T)|$, we obtain that there exists a dismembering set $F$ for $(T-v,X)$ with $|F|\leq 5|X|$. It follows directly from the definition of dismembering sets that $F$ is also a dismembering set for $(T,X)$, a contradiction. 
    \end{proof}
    Next, let $F_1=E_{\geq 3}(T)$. By \Cref{treeedge} and  \Cref{leafx}, we have that $|F_1|\leq 3|E_1(T)|\leq 3|X|$. Observe that $T\setminus F_1$ is a collection of paths. Let $F_2$ contain all edges of $E(T)\setminus F_1$ that are incident to a vertex of $X$. As $d_{T\setminus F_1}(v)\leq 2$ holds for all $v \in V(T)$, we obtain that $|F_2|\leq 2|X|$. Let $F=F_1 \cup F_2$ and observe that $|F|=|F_1|+|F_2|\leq 3|X|+2|X|=5|X|$. Further observe that every vertex of $X$ is isolated in $G\setminus F$ and every component of $G\setminus F$ is either a single vertex or a path that only contains vertices that are of degree 2 in $T$ . It follows that $F$ is a dismembering set for $(G,X)$. 

For the algorithmic part, observe  that $F$ can indeed be computed in polynomial time. To this end, we recursively delete leaves in $V(T)\setminus X$ from $T$ and then follow the explicit description of $F_1$ and $F_2$.
\end{proof}
We are now ready to conclude \Cref{edgedis}, which we restate here for convenience.
\edgedis*
\begin{proof}
    Let $T_1,\ldots,T_q$ be the components of $G\setminus F_0$. By assumption, for $i \in [q]$, we have that $T_i$ is a tree. Next, for $i \in [q]$, let $X_i=V(F_0)\cap V(T_i)$. Observe that $\sum_{i \in [q]}|X_i|=|V(F_0)|\leq 2|F_0|.$ By \Cref{vertexdis}, for $i \in [q]$, there exists a dismembering set $F_i$ for $(T_i,X_i)$ with $|F_i|\leq 5|X_i|$. Let $F=F_1\cup \ldots \cup F_q$. It is easy to see that $F$ is a dismembering set for $(G\setminus  F_0,V(F_0))$. Further, we have $|F|=\sum_{i \in [q]}|F_i|\leq \sum_{i \in [q]}5|X_i|=5\sum_{i \in [q]}|X_i|\leq 10|F_0|.$ Hence $F$ has the desired properties. Further, it follows directly from the algorithmic part of \Cref{vertexdis} that $F$ can be computed in polynomial time.
\end{proof}

In the following, we conclude \Cref{reddismem} from \Cref{edgedis}.  Given an instance $(G,w)$ of WAMMRO, we first show how to obtain a small set of partial orientations corresponding to dismembered instances of WAMMRO which maintains all important properties of $(G,w)$.

\begin{lemma}\label{uhkijiop}
    Let $(G,w)$ be an instance of WAMMRO. Then, in time $2^{O(k)}n^{O(1)}$, we can compute a collection $\mathcal{G}$ of $2^{O(k)}$ partial orientations of $G$ such that $(G',w)$ is a dismembered instance of WAMMRO with $|A(G')|\leq 11 k$ for all $G' \in \mathcal{G}$ and every orientation of $G$ is also an orientation of $G'$ for some $G' \in \mathcal{G}$.
\end{lemma}
\begin{proof}
    Let $F_0 \subseteq E(UG(G))$ be the set of edges corresponding to arcs in $A(G)$. By \Cref{edgedis}, in polynomial time, we can compute a set $F \subseteq E(G)$ with $|F|\leq 10 |F_0|$ that is dismembering for $(UG(G)\setminus F_0,V(F_0))$. Now for every possible orientation $\vec{F}$ of $F$, we let $G_{\vec{F}}$ be the unique partial orientation of $G$ in which all edges in $F$ are oriented as in $\vec{F}$ and all edges in $E(G)\setminus F$ remain unoriented. We let $\mathcal{G}$ contain $G_{\vec{F}}$ for all possible orientations $\vec{F}$ of $F$. Clearly, we have that $|\mathcal{G}|=2^{|F|}=2^{O(k)}$ and $\mathcal{G}$ can be computed in time $2^{O(k)}n^{O(1)}$.
    
    Now consider $G_{\vec{F}}$ for some orientation $\vec{F}$ of $F$. Clearly, we have $|A(G_{\vec{F}})|=|F_0|+|F|\leq 11k$. Observe that every component $C$ of $G_{\vec{F}}\setminus A(G_{\vec{F}})$ is also a component of $G\setminus (F \cup F_0)$. As $F$ is a dismembering set for $(G\setminus F_0,V(F_0))$, we obtain that $C$ either contains at most one vertex of $V(F \cup F_0)$ or $C$ contains no vertex in $V(F_0)$ and exactly two vertices in $V(F)$ both of which are incident to exactly one edge in $F$. In the first case, we obtain that $C$ contains at most one vertex of $V(A(G_{\vec{F}}))$ and in the latter case that $C$ contains exactly two vertices of $V(A(G_{\vec{F}}))$ and each of them is incident to exactly one arc of $A(G_{\vec{F}})$. It follows that $(G_{\vec{F}},w)$ is a dismembered instance of WAMMRO.

    Finally, let $\vec{G}$ be an orientation of $G$ and let $\vec{F}$ be the orientation of $F$ in $\vec{G}$. Then $\vec{G}$ is also an orientation of $G_{\vec{F}}$.
\end{proof}
We are now ready to prove \Cref{reddismem} for connected instances of WAMMRO.
\begin{lemma}\label{setrzdtufzguh}
    Let $g:\mathbb{Z}_{\geq 0}\times \mathbb{Z}_{\geq 0}\times [0,1]\rightarrow \mathbb{Z}_{\geq 0}$ be a fixed function and suppose that there exists an algorithm $A$ that computes a $(1-\epsilon)$-optimal orientation for every connected, dismembered instance of WAMMRO and every $\epsilon>0$ in time $g(n,k,\epsilon)$. Then, for every $\epsilon>0$, a $(1-\epsilon)$-optimal orientation for every connected instance of WAMMRO can be computed in time $f(k)g(n,k',\epsilon)n^{O(1)}$ for some $k'=O(k)$.
 \end{lemma}
 \begin{proof}
     We first compute a collection $\mathcal{G}$ of $2^{O(k)}$ partial orientations of $G$ such that $(G',w)$ is a dismembered instance of WAMMRO with $|A(G')|\leq k'$ for all $G' \in \mathcal{G}$ for some $k' \leq 11k$ and every orientation of $G$ is also an orientation of $G'$ for some $G' \in \mathcal{G}$. Now, for every $G' \in \mathcal{G}$, we compute a $(1-\epsilon)$-optimal orientation $\vec{G'}$ for $(G',w)$ using $A$. Observe that for any $G' \in \mathcal{G}$, by assumption, we have that $\vec{G'}$ is also an orientation of $G$. We maintain the orientation $\vec{G}_0$ of $G$ found during this procedure that maximizes $R(\vec{G}_0,w)$.

     In order to see that $\vec{G}_0$ is a $(1-\epsilon)$-optimal orientation for $(\vec{G},w)$, let $\vec{G}$ be an arbitrary orientation of $G$. Then there exists some $G' \in \mathcal{G}$ such that $\vec{G}$ is an orientation of $G'$. Let $\vec{G'}$ be the $(1-\epsilon)$-optimal orientation for $(G',w)$ that we computed during our algorithm when considering $G'$. Then, by construction, we have $R(\vec{G}_0,w)\geq R(\vec{G'},w)\geq (1-\epsilon)R(\vec{G},w)$ and so $\vec{G}_0$ is a $(1-\epsilon)$-optimal orientation for $(\vec{G},w)$.

     For the running time, first observe that $\mathcal{G}$ can be computed in time $2^{O(k)}n^{O(1)}$ by \Cref{uhkijiop}. Next, for every $G' \in \mathcal{G}$, we have that $\vec{G'}$ can be computed in time $g(n,k',\epsilon)$ using $A$ by assumption. Finally, as $|\mathcal{G}|=2^{O(k)}$, we obtain that the whole algorithm runs in time $f(k)g(n,k',\epsilon)n^{O(1)}$.
 \end{proof}

 We are now ready to prove \Cref{reddismem}, which we restate here for convenience, in its full generality.
\reddismem*

 \begin{proof}
      Let $G^1,\ldots,G^q$ be the unique collection of vertex-disjoint mixed subgraphs of $G$ such that $UG[G^1],\ldots,UG[G^q]$ are the components of $UG[G]$.

     For every $i \in [q]$, observe that $(G^i,w|V(G^i))$ is a connected instance of WAMMRO. We now compute a $(1-\epsilon)$-optimal orientation $\vec{G}^{i}_0$ of $G^{i}$. We do this for every $i \in [q]$. We now let $\vec{G}_0$ be the unique orientation of $G$ in which $e$ has the same orientation in $\vec{G}_0$ and $\vec{G}_0^{i}$ for every $e \in E(G_i)$ and every $i \in [q]$. 

     Now let $\vec{G}$ be an arbitrary orientation of $G$ and for $i \in [q]$, let $\vec{G}^{i}$ be the orientation of $G^{i}$ inherited from $\vec{G}$. Clearly, we have $\kappa_{\vec{G}}(u,v)=0$ for all $u\in V(G^{i}),v \in V(G^{j})$ and all $i,j \in [q]$ with $i \neq j$.
     This yields 
     \begin{align*}
         R(\vec{G}_0,w)&=\sum_{i \in [q]}R(\vec{G}_0^{i},w|V(G_i))\\
         &\geq \sum_{i \in [q]}(1-\epsilon)R(\vec{G}^{i},w|V(G_i))\\
         &=(1-\epsilon)\sum_{i \in [q]}R(\vec{G}^{i},w|V(G_i))\\
         &=(1-\epsilon)R(\vec{G},w).
     \end{align*}

     For the running time, observe that by \Cref{setrzdtufzguh}, we can compute $\vec{G}_0^{i}$ in time $f(k)g(n,k',\epsilon)n^{O(1)}$ for some $k'=O(k)$ for all $i \in [q]$. As $q \leq n$, it follows that $\vec{G}_0$ can be computed in time $f(k)g(n,k',\epsilon)n^{O(1)}$.
 \end{proof}

\subsection{Simulation sets}\label{sec43}
We here prove \Cref{rdtfguh}, which is restated below. We do this by distinguishing two cases depending on the exact interaction of the undirected component $T$ with $A(G)$. In either case, we explicitly construct the simulation set and prove that it has the desired properties. 
\rdtfguh*
\begin{proof}
We use $\overline{V}$ for $V(G)\setminus V(T)$. Recall that $T$ is a tree as $T$ is an undirected graph by definition and a mixed subgraph of $G$, which is acyclic.
    Let $P_1$ contain all  $(u,v)\in P(G)$ with $\{u,v\}\subseteq \overline{V}$, let $P_2$ contain all $(u,v)\in P(G)$ with $\{u,v\}\subseteq V(T)$, let $P_3$ contain all $(u,v)\in P(G)$ with $u \in V(T)$ and $v \in \overline{V}$, and let $P_4$ contain all $(u,v)\in P(G)$ with $u \in \overline{V}$ and $v \in V(T)$.  Observe that $(P_1,P_2,P_3,P_4)$ is a partition of $P(G)$. We now distinguish two cases depending on the exact interaction of $T$ with $A(G)$. It follows from the fact that $(G,w)$ is a dismembered instance of WAMMRO that one of the following cases occurs.

  \begin{Case}\label{c1}
       $\delta_G^-(V(T))\subseteq \delta_G^-(x_{in})$ and $\delta_G^+(V(T))\subseteq \delta_G^+(x_{out})$ for some $x_{in},x_{out}\in V(T)$.
    \end{Case}
   We first construct a mixed graph $U_1$ from $T$ by adding two vertices $z_{in}$ and $z_{out}$ and the arcs $z_{in}x_{in},x_{out}z_{out}$, and $z_{in}z_{out}$. We further obtain $U_2$ from $U_1$ by orienting all edges on the unique $x_{in}x_{out}$-path in $T$ so that we obtain a directed $x_{in}x_{out}$-path and we set $\mathcal{U}=\{U_1,U_2\}$. An illustration of $\mathcal{U}$ can be found in \Cref{drftg2}. 
   \begin{figure}[h]
    \centering
        \includegraphics[width=.8\textwidth]{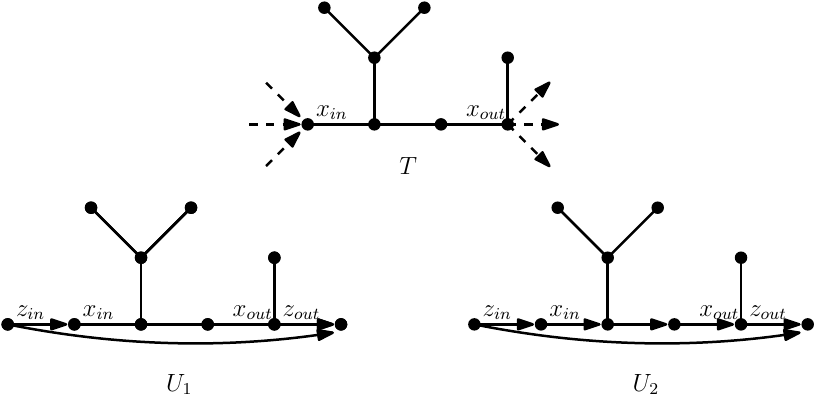}
        \caption{An example of the construction in the proof of \Cref{c1}. On top, there is an example of an undirected component $T$, where the dashed halfarcs mark the arcs in $A(G)$ incident to $V(T)$. The corresponding mixed graphs $U_1$ and $U_2$ are depicted below.}\label{drftg2}
\end{figure}Clearly, we have that $\mathcal{U}$ can be computed in polynomial time. We will show that $\mathcal{U}$ is an arboresque 3-simulation set for $((G,w),T)$. We clearly have that $U_1$ and $U_2$ are mixed 3-supergraphs of $T$. Next observe that $d_{UG(U_1)}(z_{in})=d_{U_1}^+(z_{in})=2$ and $UG(U_1\setminus \{z_{in}z_{out}\})$ is a tree. Hence $U_1$ is arboresque. A similar argument shows that $U_2$ is also arboresque. For the last property, we first need some preliminary results.

  \begin{claim}\label{summe}
       Let $\vec{G}$ be an orientation of $G$. Then 
       \begin{align*}
       \sum_{(u,v)\in P_3}\kappa_{\vec{G}}(u,v)w(u)w(v)&=\sum_{u \in V(T)}w(u)\kappa_{\vec{G}}(u,x_{out})\sum_{v \in \overline{V}}w(v)\kappa_{\vec{G}}(x_{out},v)\\
       &\text{ and } \\
       \sum_{(u,v)\in P_4}\kappa_{\vec{G}}(u,v)w(u)w(v)&=\sum_{u \in \overline{V}}w(u)\kappa_{\vec{G}}(u,x_{in})\sum_{v \in V(T)}w(v)\kappa_{\vec{G}}(x_{in},v).\end{align*}
   \end{claim}
\begin{proof}
    By symmetry, it suffices to prove the first statement.

    Consider some $(u,v)\in P_3$. We show that $\kappa_{\vec{G}}(u,v)=\kappa_{\vec{G}}(u,x_{out})\kappa_{\vec{G}}(x_{out},v)$.
    
    First suppose that $\kappa_{\vec{G}}(u,v)=1$, so there exists a directed $uv$-path $P$ in $\vec{G}$. As $A(P)$ needs to contain an arc leaving $V(T)$ and by the assumption on $x_{out}$, we obtain that $x_{out}\in V(P)$. Hence $P$ contains a directed $ux_{out}$-path and a directed $x_{out}v$-path. This yields $\kappa_{\vec{G}}(u,x_{out})=\kappa_{\vec{G}}(x_{out},v)=1$ and hence $\kappa_{\vec{G}}(u,x_{out})\kappa_{\vec{G}}(x_{out},v)=1$.

    Now suppose that $\kappa_{\vec{G}}(u,x_{out})\kappa_{\vec{G}}(x_{out},v)=1$. Clearly, we obtain $\kappa_{\vec{G}}(u,x_{out})=\kappa_{\vec{G}}(x_{out},v)=1$, so $\vec{G}$ contains a directed $ux_{out}$-path $P_1$ and a directed $x_{out}v$-path $P_2$. Then the concatenation of $P_1$ and $P_2$ is a directed $uv$-walk in $\vec{G}$, so $\kappa_{\vec{G}}(u,v)=1$. 

    Hence $\kappa_{\vec{G}}(u,v)=\kappa_{\vec{G}}(u,x_{out})\kappa_{\vec{G}}(x_{out},v)$ holds for all $(u,v)\in P_3$. This yields 

    \begin{align*}
       \sum_{(u,v)\in P_3}\kappa_{\vec{G}}(u,v)w(u)w(v)&=\sum_{(u,v)\in P_3}\kappa_{\vec{G}}(u,x_{out})\kappa_{\vec{G}}(x_{out},v)w(u)w(v)\\&=\sum_{u \in V(T)}w(u)\kappa_{\vec{G}}(u,x_{out})\sum_{v \in \overline{V}}w(v)\kappa_{\vec{G}}(x_{out},v).
     \end{align*}
\end{proof}
\begin{claim}\label{estdrzftugziuhi}
    Let $\vec{G}$ be an orientation of $G$, $u,u'\in V(T)$ and $v \in \overline{V}$. Then $\kappa_{\vec{G}}(u,v)\kappa_{\vec{G}}(v,u')=0$.
\end{claim}
\begin{proof}
Suppose otherwise, so $\vec{G}$ contains a directed $uv$-path $Q_1'$. As $u \in V(T), v\in \overline{V}$, and $\delta^+_G(V(T))\subseteq \delta^+_G(x_{out})$, we obtain that $Q_1'$ contains a directed $x_{out}v$-path. In particular, we have that $G$ contains a mixed $x_{out}v$-path $Q_1$. Similarly, we have that $G$ contains a mixed $vx_{in}$-path $Q_2$.
    Further, as $T$ is a tree, there exists a mixed $x_{in}x_{out}$-path $Q_3$ in $T$. The concatenation of $Q_1,Q_2$, and $Q_3$ is a mixed cycle in $G$, a contradiction to $G$ being acyclic.
\end{proof}
   \begin{claim}\label{dasdas}
       Let $\vec{G}$ be an orientation of $G$ and $\vec{U_1}$ an orientation of $U_1$ such that $\vec{G}$ and $\vec{U_1}$ agree on $E(T)$. Then $\kappa_{\vec{U_1}}(u,v)=\kappa_{\vec{G}}(u,v)$ holds for all $(u,v)\in P_2$.
   \end{claim}
\begin{proof}
    Let $(u,v)\in P_2$. First suppose that $\kappa_{\vec{U_1}}(u,v)=1$, so there exists a directed $uv$-path $Q$ in $\vec{U_1}$. As $z_{in}$ is a source and $z_{out}$ is a sink in $\vec{U_1}$, we obtain that $V(Q)\cap \{z_{in},z_{out}\}=\emptyset$ and so $P$ also exists in $\vec{G}$, yielding $\kappa_{\vec{G}}(u,v)=1$.
    
    Now suppose that $\kappa_{\vec{G}}(u,v)=1$, so there exists a directed $uv$-path $P$ in $\vec{G}$. If $V(P)\setminus V(T)$ contains a vertex $v'\in \overline{V}$, we obtain that $\kappa_{\vec{G}}(u,v')\kappa_{\vec{G}}(v',v)=1$, a contradiction to Claim \ref{estdrzftugziuhi}. This yields $V(P)\subseteq V(T)$ and hence $P$ also exists in $\vec{U_1}$, yielding $\kappa_{\vec{T}}(u,v)=1$. 
\end{proof}

We now give the main proof of the last property.
   To this end, let $\vec{G}_1$ be an orientation of $G$. We define the weight function $w':V(T')\rightarrow [n]_0$ by $w'(v)=w(v)$ for all $v \in V(T)$, $w'(z_{in})=\sum_{v \in \overline{V}}\kappa_{\vec{G}_1}(v,x_{in})w(v)$ and $w'(z_{out})=\sum_{v \in \overline{V}}\kappa_{\vec{G}_1}(x_{out},v)w(v)$. Observe that $w'$ is consistent with $w$. Further, by Claim \ref{estdrzftugziuhi}, we obtain that $|w'|\leq |w|$.
   
   We now choose $U=U_1$ if $\kappa_{\vec{G}_1}(x_{in},x_{out})=0$ and $U=U_2$ if $\kappa_{\vec{G}_1}(x_{in},x_{out})=1$.

   Now let $\vec{U}$ be an  orientation for $(U,w')$. We further define $\vec{G}_0=\vec{G}_1\langle \vec{U} \rangle$ and $\vec{U}_0=\vec{U}\langle \vec{G}_1 \rangle$. It follows directly by construction that $\vec{G}_0$ is an orientation of $G$ and by the definition of $U$ that $\vec{U}_0$ is an orientation of $U$.

   We need one more preliminary result relating the reachabilities in $\vec{G}_1$ and $\vec{G}_0$.
   \begin{claim}\label{p1g}
       Let $(u,v)\in P(G)$ such that one of the following holds:
       \begin{itemize}
       \item $(u,v)\in P_1$,
       \item $u=x_{out}$ and $v \in \overline{V}$,
       \item $u \in \overline{V}$ and $v=x_{in}$.
       \end{itemize}Then $\kappa_{\vec{G_1}}(u,v)\leq \kappa_{\vec{G_0}}(u,v)$.
   \end{claim}
   \begin{proof}

       Suppose that $\kappa_{\vec{G}_1}(u,v)=1$, so $\vec{G}_1$ contains a directed $uv$-path $Q$. We first consider the case that $u=x_{out}$ and $v \in \overline{V}$. As $Q$ is a directed path, by the assumption on $x_{out}$ and as $v \in \overline{V}$, we obtain that $V(Q)\cap V(T)=\{x_{out}\}$. Hence $Q$ also exists in $\vec{G_0}$, so $\kappa_{\vec{G_0}}(u,v)=1$. The statement similarly follows if $u \in \overline{V}$ and $v=x_{in}$.
       
       We may hence suppose that $(u,v)\in P_1$. If $V(Q)\cap V(T)=\emptyset$, we obtain again that $Q$ also exists in $\vec{G}_0$. We may hence suppose that $V(Q)\cap V(T)\neq \emptyset$. 
       By the assumptions on $x_{in}$ and $x_{out}$, we obtain that $Q[V(T)]$ is the unique orientation of the $x_{in}x_{out}$-path in $T$ as a directed $x_{in}x_{out}$-path. In particular, we obtain that $\kappa_{\vec{G}_1}(x_{in},x_{out})=1$, so by the choice of $U$, we have that $U=U_2$. As $\vec{U}$ is an orientation of $U$, we obtain that $Q[V(T)]$ also exists in $\vec{U}$ and hence $Q$ also exists in $\vec{G_0}$. This yields $\kappa_{\vec{G_0}}(u,v)=1$.
   \end{proof}
   We are now ready to show that $R(\vec{G_0},w)-R(\vec{G}_1,w)\geq R(\vec{U},w')-R(\vec{U}_0,w')$, which proves the statement. We first give the full calculation and then justify every single transformation.
  
   \begin{align*}
       &R(\vec{G_0},w)-R(\vec{G}_1,w)\\
       =&\left(2\sum_{v \in V(G)}{w(v)\choose 2}+\sum_{(u,v) \in P(G)}\kappa_{\vec{G_0}}(u,v)w(u)w(v)\right)\\&-\left(2\sum_{v \in V(G)}{w(v)\choose 2}+\sum_{(u,v) \in P(G)}\kappa_{\vec{G}_1}(u,v)w(u)w(v)\right)\\
       \geq &\sum_{(u,v) \in P_2\cup P_3 \cup P_4}\kappa_{\vec{G_0}}(u,v)w(u)w(v)-\sum_{(u,v) \in P_2\cup P_3 \cup P_4}\kappa_{\vec{G}_1}(u,v)w(u)w(v)\\
=&\sum_{(u,v) \in P_2}\kappa_{\vec{U}}(u,v)w(u)w(v)-\sum_{(u,v) \in P_2}\kappa_{\vec{U}_0}(u,v)w(u)w(v)\\
&+\sum_{u \in V(T)}\kappa_{\vec{G_0}}(u,x_{out})w(u)\sum_{v \in \overline{V}}\kappa_{\vec{G_0}}(x_{out},v)w(v)+\sum_{u \in \overline{V}}\kappa_{\vec{G_0}}(u,x_{in})w(u)\sum_{v \in V(T)}\kappa_{\vec{G_0}}(x_{in},v)w(v)\\
&-\sum_{u \in V(T)}\kappa_{\vec{G}_1}(u,x_{out})w(u)\sum_{v \in \overline{V}}\kappa_{\vec{G}_1}(x_{out},v)w(v)-\sum_{u \in \overline{V}}\kappa_{\vec{G}_1}(u,x_{in})w(u)\sum_{v \in V(T)}\kappa_{\vec{G}_1}(x_{in},v)w(v)\\ 
\geq&\sum_{(u,v) \in P_2}\kappa_{\vec{U}}(u,v)w'(u)w'(v)-\sum_{(u,v) \in P_2}\kappa_{\vec{U}_0}(u,v)w'(u)w'(v)\\
&+\sum_{u \in V(T)}\kappa_{\vec{G_0}}(u,x_{out})w'(u)\sum_{v \in \overline{V}}\kappa_{\vec{G}_1}(x_{out},v)w(v)+\sum_{u \in \overline{V}}\kappa_{\vec{G}_1}(u,x_{in})w(u)\sum_{v \in V(T)}\kappa_{\vec{G_0}}(x_{in},v)w'(v)\\
&-\sum_{u \in V(T)}\kappa_{\vec{G}_1}(u,x_{out})w'(u)\sum_{v \in \overline{V}}\kappa_{\vec{G}_1}(x_{out},v)w(v)-\sum_{u \in \overline{V}}\kappa_{\vec{G}_1}(u,x_{in})w(u)\sum_{v \in V(T)}\kappa_{\vec{G}_1}(x_{in},v)w'(v)\\
=&\sum_{(u,v) \in P_2}\kappa_{\vec{U}}(u,v)w'(u)w'(v)-\sum_{(u,v) \in P_2}\kappa_{\vec{U}_0}(u,v)w'(u)w'(v)\\
&+\sum_{u \in V(T)}\kappa_{\vec{U}}(u,z_{out})w'(u)w'(z_{out})+\sum_{v \in V(T)}\kappa_{\vec{U}}(z_{in},v)w'(z_{in})w'(v)\\
&-\sum_{u \in V(T)}\kappa_{\vec{U}_0}(u,z_{out})w'(u)w'(z_{out})-\sum_{v \in V(T)}\kappa_{\vec{U}_0}(z_{in},v)w'(z_{in})w'(v)\\
        =&R(\vec{U},w')-R(\vec{U}_0,w').
   \end{align*}

   The first equality is the definition of $R$.
   
The second inequality follows from \Cref{p1g}.

The third equality follows from \Cref{dasdas} and \Cref{summe}.

The fourth inequality follows from \Cref{p1g} and the fact $w'(v)=w(v)$ for all $v \in V(T)$.

The fifth equality follows from the definition of $w'$ and the fact that $\kappa_{\vec{U}_1}(z_{in},v)=\kappa_{\vec{U}_1}(x_{in},v)$ and $\kappa_{\vec{U}_1}(v,z_{out})=\kappa_{\vec{U}_1}(v,x_{out})$ holds for every $v \in V(T)$ and every orientation $\vec{U}_1$ of $U$.

The sixth equality follows from the definition of $R$ and the fact that $\kappa_{\vec{U}_1}(z_{in},z_{out})=1$ and $\kappa_{\vec{U}_1}(z_{out},z_{in})=0$ hold for every orientation $\vec{U}_1$ of $\vec{U}$.

 \begin{Case}\label{c2}
       One of $\delta_G^+(V(T))$ and $\delta_G^-(V(T))$ is empty and $|V(T)\cap V(A(G))|=2$.
    \end{Case}
    By symmetry, we may suppose that $\delta_G^+(V(T))=\emptyset$. Let $x_1,x_2$ be the two unique vertices in $V(T)$ incident to an arc of $\delta_G^-(V(T))$. 

   We now construct a mixed graph $U$ from $T$ by adding three vertices $z_{1},z_2$, and $z_3$  and the arcs $z_{1}x_{1},z_2x_2,z_3x_1$, and $z_{3}x_{2}$. We clearly have that $U$ is a mixed 3-supergraph of $T$. Next observe that $d_{UG(U)}(z_{3})=d_U^+(z_3)=2$ and $UG(U\setminus \{z_{3}x_{1}\})$ is a tree. Hence $U$ is arboresque. We set $\mathcal{U}=\{U\}$. An illustration of $\mathcal{U}$ can be found in \Cref{drftg3}. 
   \begin{figure}[h]
    \centering
        \includegraphics[width=.8\textwidth]{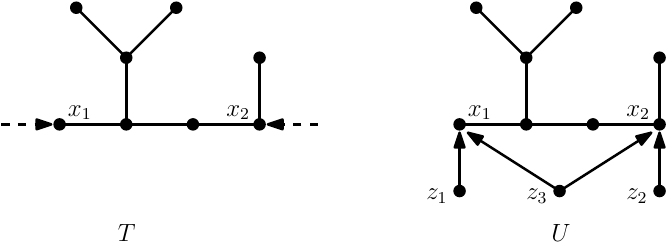}
        \caption{An example of the construction in the proof of \Cref{c2}. On the left, there is an example of an undirected component $T$, where the dashed halfarcs mark the arcs in $A(G)$ incident to $V(T)$. The corresponding mixed graph $U$ is depicted on the right.}\label{drftg3}
\end{figure}Clearly, we have that $\mathcal{U}$ can be computed in polynomial time.

   Again, before proving the last property of simulation sets, we give some preliminary results.
    \begin{claim}\label{dasdas2}
       Let $\vec{G}$ be an orientation of $G$ and $\vec{U}$ an orientation of $U$ such that $\vec{G}$ and $\vec{U}$ agree on $E(T)$. Then $\kappa_{\vec{U}}(u,v)=\kappa_{\vec{G}}(u,v)$ holds for all $(u,v)\in P_2$.
   \end{claim}
\begin{proof}
    Let $(u,v)\in P_2$. First suppose that $\kappa_{\vec{U}}(u,v)=1$, so there exists a directed $uv$-path $Q$ in $\vec{U}$. As $z_{1},z_2$, and $z_3$ are sources in $\vec{U}$, we obtain that $V(Q)\cap \{z_{1},z_{2},z_3\}=\emptyset$ and so, as $\vec{G}$ and $\vec{U}$ agree on $E(T)$, we have that $Q$ also exists in $\vec{G}$, yielding $\kappa_{\vec{G}}(u,v)=1$.
    
    Now suppose that $\kappa_{\vec{G}}(u,v)=1$, so there exists a directed $uv$-path $Q$ in $\vec{G}$. As $\vec{G}$ does not contain an arc leaving $V(T)$, we obtain that $V(Q)\subseteq V(T)$. Hence, as $\vec{G}$ and $\vec{U}$ agree on $E(T)$, it follows that $Q$ also exists in $\vec{U}$, so  $\kappa_{\vec{U}}(u,v)=1$.
\end{proof}
   \begin{claim}\label{setdrzftugzhuji}
    Let $(u,v)\in P_4$ and $\vec{G}$ an orientation of $G$. Then $$\kappa_{\vec{G}}(u,v)=\max\{\kappa_{\vec{G}}(u,x_1)\kappa_{\vec{G}}(x_1,v), \kappa_{\vec{G}}(u,x_2)\kappa_{\vec{G}}(x_2,v)\}.$$
\end{claim}
\begin{proof}
    First suppose that $\max\{\kappa_{\vec{G}}(u,x_1)\kappa_{\vec{G}}(x_1,v), \kappa_{\vec{G}}(u,x_2)\kappa_{\vec{G}}(x_2,v)\}=1$. By symmetry, we may suppose that $\kappa_{\vec{G}}(u,x_1)\kappa_{\vec{G}}(x_1,v)=1$, so $\kappa_{\vec{G}}(u,x_1)=\kappa_{\vec{G}}(x_1,v)=1$. It follows that $\vec{G}$ contains a directed $ux_1$-path $Q_1$ and a directed $x_1v$-path $Q_2$. Hence the concatenation of $Q_1$ and $Q_2$ is a directed $uv$-walk in $\vec{G}$, so $\kappa_{\vec{G}}(u,v)=1$.

    Now suppose that $\kappa_{\vec{G}}(u,v)=1$, so there exists a directed $uv$-path $Q$ in $\vec{G}$. As $u \in V(G)\setminus V(T), v \in V(T)$, and the head of every arc in $A(\vec{G})$ entering $V(T)$ is $x_1$ or $x_2$, we obtain that $\{x_1,x_2\}\cap V(Q)\neq \emptyset$. By symmetry, we may suppose that $x_1 \in V(Q)$. It follows that $Q$ contains a directed $ux_1$-path and a directed $x_1v$-path. This yields $\max\{\kappa_{\vec{G}}(u,x_1)\kappa_{\vec{G}}(x_1,v), \kappa_{\vec{G}}(u,x_2)\kappa_{\vec{G}}(x_2,v)\}\geq \kappa_{\vec{G}}(u,x_1)\kappa_{\vec{G}}(x_1,v)=1$. 
\end{proof}

We now prove the last property of simulation sets. Let $\vec{G}_1$ be an orientation of $G$.  Before defining $w'$, we need some extra definitions and observations. We first define three subsets $\overline{V_1},\overline{V_2}$, and $\overline{V_3}$ in the following way: we let $\overline{V_1}$ consist of all $v \in \overline{V}$ that satisfy $\kappa_{\vec{G}_1-\{x_2\}}(v,x_1)=1$ and $\kappa_{\vec{G}_1-\{x_1\}}(v,x_2)=0$, we let $\overline{V_2}$ consist of all $v \in \overline{V}$ that satisfy $\kappa_{\vec{G}_1-\{x_1\}}(v,x_2)=1$ and $\kappa_{\vec{G}_1-\{x_2\}}(v,x_1)=0$, and we let $\overline{V_3}$ consist of all $v \in \overline{V}$ that satisfy $\kappa_{\vec{G}_1-\{x_1\}}(v,x_2)=1$ and $\kappa_{\vec{G}_1-\{x_2\}}(v,x_1)=1$. Observe that $(\overline{V_1},\overline{V_2},\overline{V_3})$ is a subpartition of $\overline{V}$.
For $i\in [3]$, we define $P_4^{i}$ to be the set of pairs $(u,v)\in P_4$ with $u \in \overline{V_i}$.
   We now define $w':V(U)\rightarrow [n]_0$ by $w'(z_i)=w(\overline{V_i})$ for $i \in [3]$ and $w'(v)=w(v)$ for all $v \in V(T)$. Observe that $w'$ is consistent with $w$. Further, as  $(\overline{V_1},\overline{V_2},\overline{V_3})$ is a subpartition of $\overline{V}$, it follows that $|w'|\leq |w|$.

    Now let $\vec{U}$ be an orientation of $U$. We further define $\vec{G}_0=\vec{G}_1\langle \vec{U} \rangle$ and $\vec{U}_0=\vec{U}\langle \vec{G}_1\rangle$. It follows directly by construction that $\vec{G}_0$ is an orientation of $G$ and $\vec{U}_0$ is an orientation of $U$.
Again, before giving the final property of simulation sets, we need to prove some preliminary results.
    \begin{claim}\label{p1g2}
    Let $(u,v)\in P(G)$ with $\kappa_{\vec{G}_1}(u,v)> \kappa_{\vec{G_0}}(u,v)$. Then $(u,v)\in P_2 \cup P_4^{1}\cup P_4^2 \cup P_4^3$.
\end{claim}
\begin{proof}
    First consider some $(u,v)\in P_1 \cup P_3$ with $\kappa_{\vec{G}_1}(u,v)=1$, so there exists a directed $uv$-path $Q$ in $\vec{G}_1$. As $d_{\vec{G}_1}^+(V(T))\leq d_{G}^+(V(T))+d_G(V(T))=0$, we obtain that $V(Q)\cap V(T)=\emptyset$. Hence $Q$ also exists in $\vec{G}_0$. This yields $ \kappa_{\vec{G_0}}(u,v)=1$. 

    Finally consider some $(u,v)\in P_4\setminus (P_4^{1}\cup P_4^2 \cup P_4^3)$. Suppose for the sake of a contradiction that $\vec{G}_1$  contains a directed $uv$-path $Q$. As $u \in \overline{V},v \in V(T)$ and $\delta_{\vec{G}_1}^-(V(T))\subseteq \delta_{\vec{G}_1}^-(\{x_1,x_2\})$, we obtain that $V(Q)\cap \{x_1,x_2\}\neq \emptyset$. By symmetry, we may suppose that $Q$ contains a directed $ux_1$-path $Q_1$ that does not contain $x_2$. This yields $\kappa_{\vec{G}_1-\{x_2\}}(u,x_1)=1$, a contradiction to $(u,v)\in P_4\setminus (P_4^{1}\cup P_4^2 \cup P_4^3)$.
\end{proof}
\begin{claim}\label{serdft}
    Let $i \in [2],(u,v)\in P_4^{i}$ and $\vec{G}$ an orientation of $G$ that agrees with $\vec{G}_1$ on $E(G)\setminus E(T)$. Then $\max\{\kappa_{\vec{G}}(u,x_1)\kappa_{\vec{G}}(x_1,v), \kappa_{\vec{G}}(u,x_2)\kappa_{\vec{G}}(x_2,v)\}=\kappa_{\vec{G}}(x_i,v)$.
\end{claim}
\begin{proof}
By symmetry, it suffices to prove the statement for $i=1$.

 First suppose that $\kappa_{\vec{G}}(x_1,v)=1$. As $(u,v)\in P_4^1$, we have $u \in \overline{V_1}$ and so $\kappa_{\vec{G}_1-\{x_2\}}(u,x_1)=1$ holds. In particular, we have that $\vec{G}_1-\{x_2\}$ contains a directed $ux_1$-path $Q$. As $x_1$ and $x_2$ are the only vertices in $V(T)$ incident to edges entering $V(T)$ in $UG(G)$, we obtain that $V(Q)\subseteq \overline{V} \cup \{x_1\}$. As $\vec{G}_1$ and $\vec{G}$ agree on $E(G)\setminus E(T)$, it follows that $Q$ is also a directed subpath of $\vec{G}$. It follows that $\kappa_{\vec{G}}(u,x_1)=1$. We obtain $\max\{\kappa_{\vec{G}}(u,x_1)\kappa_{\vec{G}}(x_1,v), \kappa_{\vec{G}}(u,x_2)\kappa_{\vec{G}}(x_2,v)\}\geq \kappa_{\vec{G}}(u,x_1)\kappa_{\vec{G}}(x_1,v)=1$.

  We now suppose that $\max\{\kappa_{\vec{G}}(u,x_1)\kappa_{\vec{G}}(x_1,v), \kappa_{\vec{G}}(u,x_2)\kappa_{\vec{G}}(x_2,v)\}=1$. In the case that  $\kappa_{\vec{G}}(u,x_1)\kappa_{\vec{G}}(x_1,v)=1$, there is nothing to prove. It remains to consider the case that $\kappa_{\vec{G}}(u,x_2)\kappa_{\vec{G}}(x_2,v)=1$, so $\kappa_{\vec{G}}(u,x_2)=\kappa_{\vec{G}}(x_2,v)=1$. It follows that $\vec{G}$ contains a directed $ux_2$-path $Q_1$ and a directed $x_2v$-path $Q_2$. Let $x_0$ be the first vertex of $Q_1$ in $V(T)$. As $x_1$ and $x_2$ are the only vertices in $V(T)$ incident to edges entering $V(T)$ in $UG(G)$, we obtain that $x_0\in \{x_1,x_2\}$. Further, as $(u,v)\in P_4^1$, we have that $u \in \overline{V_1}$ and hence $\kappa_{\vec{G}_1-\{x_1\}}(u,x_2)=0$. As $\vec{G}$ and $\vec{G}_1$ agree on $E(G)\setminus E(T)$, we obtain that $\kappa_{\vec{G}-\{x_1\}}(u,x_2)=0$ and hence $x_0 \neq x_2$. It follows that $x_0=x_1$. In particular, we have that $Q_1$ contains a directed $x_1x_2$-path $Q_3$. Observe that the concatenation of $Q_3$ and $Q_2$ is a directed $x_1v$-walk in $\vec{G}$. It follows that $\kappa_{\vec{G}}(x_1,v)=1$.
\end{proof}
\begin{claim}\label{drztfzh}
    Let $u \in \overline{V_3}$ and $\vec{G}$ an orientation of $G$ that agrees with $\vec{G}_1$ on $E(G)\setminus E(T)$. Then $\kappa_{\vec{G}}(u,x_1)=\kappa_{\vec{G}}(u,x_2)=1$.
\end{claim}
\begin{proof}
By symmetry, it suffices to prove that $\kappa_{\vec{G}}(u,x_1)=1$. By the definition of $\overline{V_3}$, there exists a directed $ux_1$-path $Q$ in $\vec{G}_1-\{x_2\}$. As the head of every arc entering $V(T)$ in $\vec{G}_1-\{x_2\}$ is $x_1$ by the assumption on $x_1$ and $x_2$, we obtain $V(Q)\cap V(T)=\{x_1\}$. As $\vec{G}$ and $\vec{G}_1$ agree on $E(G)\setminus E(T)$, we obtain that $Q$ is also a directed subpath of $\vec{G}$. It follows that $\kappa_{\vec{G}}(u,x_1)=1$.
\end{proof}
  We are now ready to show that $R(\vec{G_0},w)-R(\vec{G}_1,w)\geq R(\vec{U},w')-R(\vec{U}_0,w')$, which proves the statement. We first give the full calculation and then justify every single transformation.
\newpage
\begin{align*}
       &R(\vec{G_0},w)-R(\vec{G}_1,w)\\
       =&\left(2\sum_{v \in V(G)}{w(v)\choose 2}+\sum_{(u,v) \in P(G)}\kappa_{\vec{G_0}}(u,v)w(u)w(v)\right)\\&-\left(2\sum_{v \in V(G)}{w(v)\choose 2}+\sum_{(u,v) \in P(G)}\kappa_{\vec{G}_1}(u,v)w(u)w(v)\right)\\
       \geq &\sum_{(u,v) \in P_2 \cup P_4}\kappa_{\vec{G_0}}(u,v)w(u)w(v)-\sum_{(u,v) \in P_2 \cup P_4}\kappa_{\vec{G}_1}(u,v)w(u)w(v)\\
=&\sum_{(u,v) \in P_2}\kappa_{\vec{G_0}}(u,v)w(u)w(v)-\sum_{(u,v) \in P_2}\kappa_{\vec{G}_1}(u,v)w(u)w(v)\\
&+\sum_{(u,v) \in P_4}\max\{\kappa_{\vec{G_0}}(u,x_1)\kappa_{\vec{G_0}}(x_1,v),\kappa_{\vec{G_0}}(u,x_2)\kappa_{\vec{G_0}}(x_2,v)\}w(u)w(v)\\ &-\sum_{(u,v) \in P_4}\max\{\kappa_{\vec{G}_1}(u,x_1)\kappa_{\vec{G}_1}(x_1,v),\kappa_{\vec{G}_1}(u,x_2)\kappa_{\vec{G}_1}(x_2,v)\}w(u)w(v)\\
=&\sum_{(u,v) \in P_2}\kappa_{\vec{G_0}}(u,v)w(u)w(v)-\sum_{(u,v) \in P_2}\kappa_{\vec{G}_1}(u,v)w(u)w(v)\\
&+\sum_{(u,v) \in P_4^1}\kappa_{\vec{G_0}}(x_1,v)w(u)w(v)+\sum_{(u,v) \in P_4^2}\kappa_{\vec{G_0}}(x_2,v)w(u)w(v)\\&+\sum_{(u,v) \in P_4^3}\max\{\kappa_{\vec{G_0}}(x_1,v),\kappa_{\vec{G_0}}(x_2,v)\}w(u)w(v)\\&-\sum_{(u,v) \in P_4^1}\kappa_{\vec{G}_1}(x_1,v)w(u)w(v)-\sum_{(u,v) \in P_4^2}\kappa_{\vec{G}_1}(x_2,v)w(u)w(v)\\&-\sum_{(u,v) \in P_4^3}\max\{\kappa_{\vec{G}_1}(x_1,v),\kappa_{\vec{G}_1}(x_2,v)\}w(u)w(v)\\
=&\sum_{(u,v) \in P_2}\kappa_{\vec{U}}(u,v)w'(u)w'(v)-\sum_{(u,v) \in P_2}\kappa_{\vec{U}_0}(u,v)w'(u)w'(v)\\
&+w'(z_1)\sum_{v \in V(T)}\kappa_{\vec{U}}(z_1,v)w'(v)+w'(z_2)\sum_{v \in V(T)}\kappa_{\vec{U}}(z_2,v)w'(v)+w'(z_3)\sum_{v \in V(T)}\kappa_{\vec{U}}(z_3,v)w'(v)\\&-w'(z_1)\sum_{v \in V(T)}\kappa_{\vec{U}_0}(z_1,v)w'(v)-w'(z_2)\sum_{v \in V(T)}\kappa_{\vec{U}_0}(z_2,v)w'(v)-w'(z_3)\sum_{v \in V(T)}\kappa_{\vec{U}_0}(z_3,v)w'(v)\\
        =&R(\vec{U},w')-R(\vec{U}_0,w').
   \end{align*}

 The first equality is the definition of $R$.
   
The second inequality follows from \Cref{p1g2}.

The third equality follows from \Cref{setdrzftugzhuji}.

The fourth equality follows from \Cref{p1g2}, \Cref{serdft}, and \Cref{drztfzh}.

The fifth equality follows from the definition of $w'$, \Cref{dasdas2} and the fact that $\kappa_{\vec{U}_1}(z_{i},v)=\kappa_{\vec{U}_1}(x_{i},v)$ for $i \in[2]$ and $\kappa_{\vec{U}_1}(z_{3},v)=\max\{\kappa_{\vec{U}_1}(x_{1},v),\kappa_{\vec{U}_1}(x_{2},v)\}$ hold for  every $v \in V(T)$ and every orientation $\vec{U}_1$ of $U$.

The sixth equality follows from the definition of $R$ and the fact that $\kappa_{\vec{U}_1}(z_i,z_j)=0$ holds for all distinct $i,j \in [3]$ and every orientation $\vec{U}_1$ of $U$.
\end{proof}
\subsection{Arboresque mixed graphs}\label{sec42}

The purpose of this section is to prove \Cref{tfzghik}. We first give a short overview of the proof.

For an instance of WAMMRO containing a mixed graph whose underlying graph is a tree, a rather simple dynamic programming algorithm is available. We root the tree at an arbitrary vertex and use a bottom-up approach, where for every vertex $v$, we store whether there exists an orientation of the mixed graph $T_v$ corresponding to the subtree rooted at $v$ meeting three conditions: the objective value inside the orientation of $T_v$ is above a certain threshold, the weight of the vertices reachable from $v$ is above a certain threshold, and the weight of the vertices from which $v$ is reachable is above a certain threshold. We store the information whether such an orientation of $T_v$ exists for all possible thresholds. Now we can recursively decide whether orientations with the desired properties exist.

For an arboresque mixed graph $G$, there exists the additional arc $rs$, which forces us to be a bit more careful about the dynamic program. First, we root the tree $UG(G \setminus \{rs\})$ at $s$. Next, we make the orientations of $T_v$ respect two additional constraints: one is represented by a binary variable indicating whether $v$ is reachable from $r$ in the orientation of $T_v$ and one upper bounds the weight of all vertices reachable from both $r$ and $v$ in the orientation of the subtree. This is necessary to avoid double counting reachabilities from $r$ when determining the objective value in the end of the algorithm. We  now restate \Cref{tfzghik} and give the full proof.
\tfzghik*

\begin{proof} As $(G,w)$ is arboresque and by symmetry, we may suppose that $A(G)$ contains an arc $rs$ such that $d_G^+(r)=d_{UG(G)}(r)=2$ and the underlying graph $T$ of $G\setminus \{rs\}$ is a tree. Observe that $r$ is a leaf of $T$. In the following, we consider $T$ to be rooted at $s$.
    For every $v \in V(G)$, we let $T_v$ be the unique mixed subgraph of $G$ whose underlying graph is the subtree of $T$ rooted at $v$. We say that some $v \in V(G)$ is {\it special} if it is on the unique path from $r$ to $s$ in $T$. We now start by running a dynamic program. In the following, for the sake of simplicity, we use $Q$ for $V(G)\times\{True, False\}\times [n]_0\times [n]_0 \times [n^2]_0\times [n]_0$. For $v \in V(G)$, we use $Q_{v}$ for the set of all $q \in Q$ whose first entry is $v$ and for $V'\subseteq V(G)$, we use $Q_{V'}$ for $\bigcup_{v \in V'}Q_v$.
    
    For every $q=(v,b,k_{in},k_{out},k_{reach},k_{shared})\in Q$, we wish to compute a boolean value $z[q]$ such that $z[q]=True$ if and only if there exists an orientation $\vec{T}_v$ of $T_v$ that satisfies the following conditions:
    \begin{itemize}
        \item $b=True$ if and only if $v$ is reachable from $r$ in $\vec{T}_v$,
        \item $w(In_{\vec{T}_v}(v))\geq k_{in}$,
        \item $w(Out_{\vec{T}_v}(v))\geq k_{out}$,
        \item $R(\vec{T}_v,w)\geq k_{reach}$,
        \item $w(Out_{\vec{T}_v}(v)\cap Out_{\vec{T_v}}(r))\leq k_{shared}$.
    \end{itemize}
    We say that the assignment of $z[q]$ is {\it correct} if this condition is satisfied. Further, if $z[q]=True$, we say that $\vec{T_v}$ {\it certifies} this assignment. Throughout this dynamic program, we show how to decide the value of $z[q]$ for $q \in Q$. While we do not explicitly construct the certifying orientations, it is not difficult to trace them through the program and we assume that they are available.
    
    The main technical part of our dynamic program is contained in the following result. \begin{claim}\label{asetdrzklö}
        Let $q=(v,b^*,k^*_{in},k^*_{out}, k^*_{reach},k^*_{shared}) \in Q_v$ for some $v \in V(G)$ and suppose that $z[q']$ has already been correctly computed for all $q' \in Q_{V(T_v)\setminus\{v\}}$. Then $z[q]$ can be computed in polynomial time.
    \end{claim}
    \begin{proof}

 We distinguish several cases depending on the properties of $v$.
 \setcounter{Case}{0}
\begin{Case}
    $v=r$.
\end{Case} We set $z[q]=True$ if all of the following conditions hold:

 \begin{itemize}
        \item $b^*=True$,
        \item $w(v)\geq k^*_{in}$,
        \item $w(v)\geq k^*_{out}$,
        \item $2{w(v) \choose 2}\geq k^*_{reach}$,
        \item $w(v)\leq k^*_{shared}$.
    \end{itemize}
Otherwise, we set $z[q]=False$.
    Observe that the only orientation of $T_v$ is $T_v$. It is hence easy to see that the assignment of $z[q]$ is correct. Further, the conditions can clearly be checked in polynomial time.
\begin{Case}
    $v$ is a leaf of $T$ distinct from $r$ and $s$.
\end{Case}
    We set $z[q]=True$ if all of the following conditions hold:

 \begin{itemize}
        \item $b^*=False$,
        \item $w(v)\geq k^*_{in}$,
        \item $w(v)\geq k^*_{out}$,
        \item $2{w(v) \choose 2}\geq k^*_{reach}$,
        \item $0\leq k^*_{shared}$.
    \end{itemize}
Otherwise, we set $z[q]=False$.
    Observe that the only orientation of $T_v$ is $T_v$. It is hence easy to see that the assignment of $z[q]$ is correct. Further, the conditions can clearly be checked in polynomial time.
\begin{Case}
    $v$ is not a leaf of $T$ or $v=s$.
\end{Case}
Let $x_1,\ldots,x_\mu$ be an ordering of the children of $v$ in $T$ such that, if $v$ is special, then $x_1$ is the unique special child of $v$ in $T$. 

In order to correctly assign $z[q]$, we need to run a second dynamic program. To this end, for $i \in [\mu]$, let $T_i=T_{x_i}$ and let $U_i=T[V(T_{1})\cup \ldots V(T_i)\cup \{v\}]$. Observe that $U_\mu=T_v$. 

We use $P$ for $[\mu]\times\{True, False\}\times [n]_0\times [n]_0 \times [n^2]_0\times [n]_0$. Further, for $i \in [\mu]$, we use $P_i$ for the elements of $P$ whose first entry is $i$. Now, for every $p=(i,b,k_{in},k_{out},k_{reach},k_{shared})$, we wish to compute a boolean value $y[p]$ such that $y[p]=True$ if and only if there exists an orientation $\vec{U}_i$ of $U_i$ that satisfies the following conditions:
    \begin{itemize}
        \item $b=True$ if and only if $v$ is reachable from $r$ in $\vec{U_i}$,
        \item $w(In_{\vec{U_i}}(v))\geq k_{in}$,
        \item $w(Out_{\vec{U_i}}(v))\geq k_{out}$,
        \item $R(\vec{U}_i,w)\geq k_{reach}$,
        \item $w(Out_{\vec{U_i}}(v)\cap Out_{\vec{U_i}}(r))\leq k_{shared}$.
    \end{itemize}

    We first show how to compute $y[p]$ for $p \in P_1$. We initialize $y[p]=False$. Consider some $q'=(x_1,b' , k'_{in}k'_{out},k'_{reach},k'_{shared})\in Q_{x_1}$.

    We first set $y[p]=True$ if all of the following conditions are satisfied:
    \begin{itemize}
    \item $z[q']=True$,
    \item $E(G)$ contains the edge $vx_1$ or $A(G)$ contains the arc $vx_1$,
    \item $b=False$,
        \item $w(v)\geq k_{in}$,
        \item $w(v)+k'_{out}\geq k_{out}$,
        \item $2{w(v)\choose 2}+w(v)k'_{out}+k'_{reach}\geq k_{reach}$,
        \item $k'_{shared}\leq k_{shared}$.
    \end{itemize}

    We further set $y[p]=True$ if all of the following conditions are satisfied:
    \begin{itemize}
    \item $z[q']=True$,
    \item $E(G)$ contains the edge $vx_1$ or $A(G)$ contains the arc $x_1v$,
    \item $b=b'$,
        \item $w(v)+k'_{in}\geq k_{in}$,
        \item $w(v)\geq k_{out}$,
        \item $2{w(v)\choose 2}+w(v)k'_{in}+k'_{reach}\geq k_{reach}$,
        \item $b=False$ or $w(v)\leq k_{shared}$.
    \end{itemize}

    We do this for every $q' \in Q_{x_1}$ and keep the assignment of $y[p]$ after the last one of these iterations.

    This is assignment is justified in the following way. Let $\vec{U}_1$ be an orientation of $U_1$ and $\vec{T}_1$ be the inherited orientation of $T_1$. If $vx_1\in A(\vec{U}_1)$, then all of the following hold:
    \begin{itemize}
        \item $v$ is not reachable from $r$ in $\vec{U}_1$,
        \item $w(In_{\vec{U}_1}(v))=w(v)$,
        \item $w(Out_{\vec{U}_1}(v))=w(v)+w(Out_{\vec{T}_1}(x_1))$,
        \item $R(\vec{U}_1,w)=R(\vec{T}_1,w)+w(v)w(Out_{\vec{T}_1}(x_1))+2{w(v)\choose 2}$,
        \item $w(Out_{\vec{U}_1}(v)\cap Out_{\vec{U}_1}(r))=w(Out_{\vec{T}_1}(x_1)\cap Out_{\vec{T}_1}(r))$.
    
    \end{itemize}

     If $x_1v\in A(\vec{U}_1)$, then all of the following hold:
    \begin{itemize}
        \item $v$ is reachable from $r$ in $\vec{U}_1$ if and only if  $x_1$ is reachable from $r$ in $\vec{T}_1$,
        \item $w(In_{\vec{U}_1}(v))=w(In_{\vec{T}_1}(x_1))+w(v)$,
        \item $w(Out_{\vec{U}_1}(v))=w(v)$,
        \item $R(\vec{U}_1,w)=R(\vec{T}_1,w)+w(v)w(In_{\vec{T}_1}(x_1))+2{w(v)\choose 2}$,
        \item $w(Out_{\vec{U}_1}(v)\cap Out_{\vec{U}_1}(r))=w(v)$ if $v$ is reachable from $r$ in $\vec{U}_1$, and $w(Out_{\vec{U}_1}(v)\cap Out_{\vec{U}_1}(r))=0$, otherwise.
    
    \end{itemize}
In order to compute $y[p]$, we need to make two checks for every $q'\in Q_{x_1}$ and each of these checks can clearly be carried out in polynomial time. As $|Q_{x_1}|=O(n^5)$ by definition, we obtain that $y[p]$ can be computed in polynomial time.

    In the following, we suppose that $p\in P_i$ for some $i \geq 2$. We may suppose that $y[p']$ has been correctly computed for all $p'\in P_{i-1}$.

     We initialize $y[p]=False$. Now consider some $p'=(i-1,b',k'_{in},k'_{out},k'_{reach},k'_{shared})\in P_{i-1}$ and some $q'=(x_i,b'',k''_{in},k''_{out},k''_{reach},k''_{shared})\in Q_{x_i}$.

    We first set $y[p]=True$ if all of the following conditions are satisfied:
    \begin{itemize}
    \item $y[p']=z[q']=True$,
    \item $E(G)$ contains the edge $vx_i$ or $A(G)$ contains the arc $vx_i$,
    \item $b=b'$,
          \item $k'_{in}\geq k_{in}$,
        \item $k'_{out}+k''_{out}\geq k_{out}$,
        \item $k'_{reach}+k'_{in}k''_{out}+k''_{reach}\geq k_{reach}$,
        \item  $k'_{shared}+bk''_{out}\leq k_{shared}$.
    \end{itemize}

 We next set $y[p]=True$ if all of the following conditions are satisfied:
    \begin{itemize}
    \item $y[p']=z[q']=True$,
    \item $E(G)$ contains the edge $x_iv$ or $A(G)$ contains the arc $x_iv$,
    \item $b=b'$,
          \item $k'_{in}+k''_{in}\geq k_{in}$,
        \item $k'_{out}\geq k_{out}$,
        \item $k'_{reach}+k'_{out}k''_{in}+k''_{reach}\geq k_{reach}$,
        \item  $k'_{shared}\leq k_{shared}$.
    \end{itemize}
 We do this for every $p'\in P_{i-1}$ and $q' \in Q_{x_i}$ and keep the assignment of $y[p]$ after the last one of these iterations.

    This assignment is justified in the following way. Let $\vec{U}_i$ be an orientation of $U_i$, let $\vec{U}_{i-1}$ be the inherited orientation of $U_{i-1}$, and let $\vec{T}_i$ be the inherited orientation of $T_i$. First observe that, as no vertex in $\{x_2,\ldots,x_\mu\}$ is special, we have that $v$ is reachable from $r$ in $\vec{U}_i$ if and only if $v$ is reachable from $r$ in $\vec{U}_{i-1}$. 

Next, if $vx_i\in A(\vec{U}_i)$, then all of the following hold:
    \begin{itemize}
       
        \item $w(In_{\vec{U}_i}(v))=w(In_{\vec{U}_{i-1}}(v))$,
        \item $w(Out_{\vec{U}_i}(v))=w(Out_{\vec{U}_{i-1}}(v))+w(Out_{\vec{T}_i}(x_i))$,
        \item $R(\vec{U}_i,w)=R(\vec{U}_{i-1},w)+w(In_{\vec{U}_{i-1}}(v))w(Out_{\vec{T}_i}(x_i))+R(\vec{T}_i,w)$,
        \item $w(Out_{\vec{U}_i}(v)\cap Out_{\vec{U}_i}(r))=w(Out_{\vec{U}_{i-1}}(v)\cap Out_{\vec{U}_{i-1}}(r))+w(Out_{\vec{T}_i}(x_i))$ if $v$ is reachble from $r$ in $\vec{U}_i$, and $w(Out_{\vec{U}_i}(v)\cap Out_{\vec{U}_i}(r))=w(Out_{\vec{U}_{i-1}}(v)\cap Out_{\vec{U}_{i-1}}(r))$, otherwise.
    
    \end{itemize}

Further, if $x_iv\in A(\vec{U}_i)$, then all of the following hold:
    \begin{itemize}
       
        \item $w(In_{\vec{U}_i}(v))=w(In_{\vec{U}_{i-1}}(v))+w(In_{\vec{T}_{i}}(x_i))$,
        \item $w(Out_{\vec{U}_i}(v))=w(Out_{\vec{U}_{i-1}}(v))$,
        \item $R(\vec{U}_i,w)=R(\vec{U}_{i-1},w)+w(Out_{\vec{U}_{i-1}}(v))w(In_{\vec{T}_i}(x_i))+R(\vec{T}_i,w)$,
        \item $w(Out_{\vec{U}_i}(v)\cap Out_{\vec{U}_i}(r))=w(Out_{\vec{U}_{i-1}}(v)\cap Out_{\vec{U}_{i-1}}(r))$.   
    \end{itemize}

In order to compute $y[p]$, we need to make two checks for every $p' \in P_{i-1}$ and every $q'\in Q_{x_1}$. Clearly,  each of these checks can be carried out in polynomial time. As $|P_{i-1}|=O(n^5)$ and $|Q_{x_1}|=O(n^5)$ by definition, we obtain that $y[p]$ can be computed in polynomial time. 

We now compute $y[p]$ for every $p \in P$. As $|P|=O(n^6)$ by definition and $y[p]$ can be computed in polynomial time for every $p \in P$, we can make $y[p]$ available for all $p \in P$ in polynomial time. Finally, we set $z[q]=y[\mu,b^*,k^*_{in},k^*_{out},k^*_{reach},k^*_{shared}]$. This assignment is correct by definition. This finishes the proof of the claim.
      \end{proof}

We now compute $z[q]$ for every $q \in Q$. As $|Q|=O(n^6)$ by construction and $z[q]$ can be computed in polynomial time for every $q \in Q$ by \Cref{asetdrzklö}, we can compute all these values in polynomial time. In particular, we may suppose that $z[q]$ is available for all $q \in Q_s$.

 We now compute the largest $K$ for which there exists some $q=(s,b,k_{in},k_{out},k_{reach},k_{shared})\in Q_s$ with $z[q]=True$ and $k_{reach}+w(r)(k_{out}-k_{shared})\geq K$. This is justified by the fact that for any orientation $\vec{G}$ of $G$, due to the choice of $r$ and $s$, we have 
\begin{align*}
R(\vec{G},w)&=R(\vec{G}\setminus \{rs\},w)+w(r)w(Out_{\vec{G}\setminus \{rs\}}(s)\setminus Out_{\vec{G}\setminus \{rs\}}(r))\\&=R(\vec{G}\setminus \{rs\},w)+w(r)(w(Out_{\vec{G}\setminus \{rs\}}(s))-w(Out_{\vec{G}\setminus \{rs\}}(s)\cap Out_{\vec{G}\setminus \{rs\}}(r))) .
   \end{align*} 
We further use the fact that $R(\vec{G},w)\leq n^2$ holds by definition for all orientations $\vec{G}$ of $G$. 
Clearly, these checks can be carried out in polynomial time as $|Q_s|=O(n^5)$. Moreover, an orientation certifying $z[q]$ for the corresponding $q$ has the desired properties.
\end{proof}
\subsection{Main lemmas}\label{secml}
This section is dedicated to proving \Cref{rdztghok} and its analogue taking into account approximate solutions that we need for the proof of \Cref{main3}. In order to state this second result, we need one more definition.  

Given some $\epsilon>0$, an instance $(G,w)$ of WAMMRO and an undirected component $T$ of $G$, we say that a set $\mathcal{T}$ of orientations of $\mathcal{T}$ is an {\it $\epsilon$-optimal replacement set} for $((G,w),T)$ if for every orientation $\vec{G}$ of $G$, there exists some $\vec{T} \in \mathcal{T}$ such that $R(\vec{G}\langle\vec{T}\rangle,w)\geq R(\vec{G},w)-\epsilon |w|^2$.
We are now ready to state this result.
\begin{restatable}
    {lemma}{rdztghgsfok}\label{rdztghgsfok}
    Let $\epsilon>0$, let $(G,w)$ be a dismembered instance of WAMMRO and let $T$ be an undirected component of $G$. Then, in polynomial time, we can compute an $\epsilon$-optimal replacement set $\mathcal{T}$ for $((G,w),T)$ of size $f(\epsilon)$. 
\end{restatable}

We now give the proofs for \Cref{rdztghok}
and \Cref{rdztghgsfok}. We start by proving \Cref{rdztghok}, which we restate here for convenience.
\rdztghok*
\begin{proof}
We initialize  $\mathcal{T}=\emptyset$.
    By \Cref{rdtfguh}, in polynomial time, we can compute a 3-simulation set $\mathcal{U}$ for $((G,w),T)$ of size at most 2. Now consider some $U \in \mathcal{U}$ and let $w':V(U)\rightarrow [n]_0$ be a weight function with $|w'|\leq |w|$ that is consistent with $w$. As $U \in \mathcal{U}$ is arboresque, by \Cref{tfzghik}, in polynomial time, we can compute an orientation $\vec{U}_{w'}$ of $U$ that maximizes $R(\vec{U}_{w'},w')$. Let $\vec{T}_{U,w'}$ be the inherited orientation of $T$. We now add $\vec{T}_{U,w'}$ to $\mathcal{T}$. We do this for every $U \in \mathcal{U}$ and every weight function $w':V(U)\rightarrow [n]_0$ with $|w'|\leq |w|$ that is consistent with $w$. The final assignment of $\mathcal{T}$ is obtained after the last one of these operations.
    
    Observe that after the last of these operations, we have that $|\mathcal{T}|=O(n^3)$ as $ |\mathcal{U}|\leq 2$ and there are $O(n^3)$ choices for $w'$ as $w'$ needs to be consistent with $w$ and $|w'|\leq |w|$ needs to be satisfied. Moreover, as $\vec{T}_{U,w'}$ can be computed in polynomial time for every $U \in \mathcal{U}$ and every considered function $w'$,  $|\mathcal{U}|\leq 2$, and there are only $O(n^3)$ choices for $w'$, we obtain that $\mathcal{T}$ can be computed in polynomial time.

   We now show that $\mathcal{T}$ is an optimal replacement set. Consider an orientation $\vec{G}$ of $G$. As $\mathcal{U}$ is a simulation set for $((G,w),T)$, there exists some $U \in \mathcal{U}$ and a function $w':V(U)\rightarrow [n]_0$ that is consistent with $w$ with $|w'|\leq |w|$ such that for every orientation $\vec{U}$ for $(U,w')$, we have that $\vec{G}\langle \vec{U}\rangle$ is an orientation of $G$, $\vec{U}\langle \vec{G}\rangle$ is an orientation of $U$ and $R(\vec{G}\langle \vec{U}\rangle,w)-R(\vec{G},w) \geq R(\vec{U},w')-R(\vec{U}\langle \vec{G}\rangle,w')$. By the optimality of $\vec{U}_{w'}$, we obtain $R(\vec{G}\langle \vec{T}_{U,w'}\rangle,w)=R(\vec{G}\langle \vec{U}_{w'}\rangle,w)\geq R(\vec{G},w)+R(\vec{U}_{w'},w')-R(\vec{U}_{w'}\langle \vec{G}\rangle,w')\geq R(\vec{G},w)$. As $\vec{T}_{U,w'}\in \mathcal{T}$ by construction, the statement follows.
\end{proof}
We still need to prove \Cref{rdztghgsfok}. The idea is, in comparison to the the proof of \Cref{rdztghok}, not to try all possible weight functions, but only a small set of weight functions such that every possible weight function is close to one in the set. In order to make this work, we need two preliminary results. The first one shows that we can indeed choose a small set of functions that is sufficiently close to all possible functions.
 \begin{proposition}\label{drzftugziuhi}
     Let $\epsilon>0,X$ a finite set and $\omega$ a positive integer. Then, in time $f(\epsilon,|X|)\omega^{O(1)}$, we can compute a set $\mathcal{C}$ of $f(\epsilon,|X|)$ mappings $c:X\rightarrow [\omega]_0$ such that for every mapping $c':X\rightarrow [\omega]_0$, there exists some $c \in \mathcal{C}$ with $|c-c'|\leq \epsilon \omega$.
 \end{proposition}
 \begin{proof}
     Let $k$ be the largest integer such that $k \epsilon \leq |X|$. Now we construct the set $S \subseteq [\omega]_0$ that contains for every $i \in \{0,\ldots,k\}$ the integer $\lceil \frac{i \epsilon \omega}{|X|} \rceil$. Now for every $z \in S^X$, we let $\mathcal{C}$ contain the mapping $c_z:X \rightarrow [\omega]_0$ where $c_z(x)=z_x$ for all $x \in X$. 

     We now show that $\mathcal{C}$ has the desired properties. First, it is easy to see that $\mathcal{C}$ can be computed in time $f(\epsilon,|X|)\omega^{O(1)}$. Next, observe that $k\leq \frac{|X|}{\epsilon}$ and hence $|\mathcal{C}|=|S^X|\leq (k+1)^{|X|}=f(\epsilon,|X|)$.

     Finally, consider some mapping $c':X \rightarrow [\omega]_0$. For every $x \in X$, let $i_x$ be the largest integer such that $\frac{i_x \epsilon \omega}{|X|} \leq c'(x)$ and let $c:X \rightarrow [\omega]_0$ be defined by $c(x)=\lceil \frac{i_x \epsilon \omega}{|X|} \rceil$ for all $x \in X$. It follows by construction that $c \in \mathcal{C}$. 

     As $c'(x)$ is an integer for all $x \in X$ and by construction, we obtain
     \begin{align*}
         |c-c'|&=\sum_{x \in X}|c'(x)-c(x)|\\
         &=\sum_{x \in X}c'(x)-c(x)\\
         &\leq \sum_{x \in X}\frac{(i_x+1)\epsilon \omega}{|X|}-\frac{i_x\epsilon \omega}{|X|}\\
         &\leq \sum_{x \in X} \frac{\epsilon \omega}{|X|}\\
         &=\epsilon \omega.
     \end{align*}
 \end{proof}
 The next result is needed to show that slightly altering the weight function of an instance $(G,w)$ of WAMMRO indeed only has a small effect on the quality of any particular orientation of $G$.
 \begin{lemma}\label{tcvzbun}
     Let $D$ be a digraph and $w,w':V(D)\rightarrow [\omega]_0$ be functions for some positive integer $\omega$. Then $R(D,w)-R(D,w')\leq 2|w+w'||w-w'|$.
 \end{lemma}
 \begin{proof}
 Throughout the proof, we use the mappings $w^+,w^-:V(D)\rightarrow [\omega]_0$ defined by $w^+(u)=\max\{w(u),w'(u)\}$ and $w^-(u)=\min\{w(u),w'(u)\}$ for all $u \in V(D)$. Observe that $|w^++w^-|=|w+w'|$ and $|w^+-w^-|=|w-w'|$. 
     As $w(v)\geq 0,w'(v)\geq 0$ and $w^+(v)-w^-(v)\geq 0$ hold for all $v \in V(D)$, we next have    \begin{align*}
         &\sum_{v \in V(D)}{w(v)\choose 2}-\sum_{v \in V(D)}{w'(v)\choose 2}\\=&
         \frac{1}{2}\sum_{v \in V(D)}w(v)(w(v)-1)-w'(v)(w'(v)-1)\\
         \leq &  \frac{1}{2}\sum_{v \in V(D)}w^+(v)(w^+(v)-1)-w^-(v)(w^-(v)-1)\\\leq &   \frac{1}{2}\sum_{v \in V(D)}w^+(v)w^+(v)-w^-(v)w^-(v)
         \\ \leq &   \frac{1}{2}\sum_{v \in V(D)}w^+(v)(w^+(v)-w^-(v))+\frac{1}{2}\sum_{v \in V(D)}w^-(v)(w^+(v)-w^-(v))\\
         \leq & \frac{1}{2}\sum_{v \in V(D)}w^+(v)\sum_{v \in V(D)}(w^+(v)-w^-(v))+\frac{1}{2}\sum_{v \in V(D)}w^-(v)\sum_{v \in V(D)}(w^+(v)-w^-(v))\\
         \leq & \frac{1}{2}|w^+||w^+-w^-|+\frac{1}{2}|w^-||w^+-w^-|\\
         \leq & |w+w'||w-w'|.
     \end{align*}
     We further have 
     \begin{align*}
        & \sum_{(u,v) \in P(D)}\kappa_D(u,v)w(u)w(v)- \sum_{(u,v) \in P(D)}\kappa_D(u,v)w'(u)w'(v)\\=&\sum_{(u,v) \in P(D)}\kappa_D(u,v)(w(u)w(v)-w'(u)w'(v))\\
         \leq &\sum_{(u,v) \in P(D)}\kappa_D(u,v)(w^+(u)w^+(v)-w^-(u)w^-(v))\\
         \leq &\sum_{(u,v) \in P(D)}w^+(u)w^+(v)-w^-(u)w^-(v)\\
         =& \sum_{(u,v) \in P(D)}w^+(u)(w^+(v)-w^-(v))+\sum_{(u,v) \in P(D)}w^-(v)(w^+(u)-w^-(u))\\
         \leq& \sum_{u\in V(D)}w^+(u)\sum_{v\in V(D)-u}(w^+(v)-w^-(v))+\sum_{v \in V(D)}w^-(v)\sum_{u\in V(D)\setminus \{v\}}(w^+(u)-w^-(u))\\
         \leq& |w^+||w^+-w^-|+|w^-||w^+-w^-|\\
        = &  |w+w'||w-w'|.
     \end{align*}
     By the definition of $R$, this yields 
     \begin{align*}
         &R(D,w)-R(D,w')\\=&2\sum_{v \in V(D)}{w(v)\choose 2}+ \sum_{(u,v) \in P(D)}\kappa_D(u,v)w(u)w(v)\\&-(2\sum_{v \in V(D)}{w'(v)\choose 2}+ \sum_{(u,v) \in P(D)}\kappa_D(u,v)w'(u)w'(v))\\
         \leq& 2 |w+w'||w-w'|.
     \end{align*}
 \end{proof}
We are now ready to conclude \Cref{rdztghgsfok}, which we restate here for convenience.
\rdztghgsfok*
\begin{proof}
    By \Cref{rdtfguh}, in polynomial time, we can compute an arboresque 3-simulation set of size at most 2 for $((G,w),T)$. Now consider some $U \in \mathcal{U}$ and let $X=V(U)\setminus V(T)$. Observe that $|X|\leq 3$ as $U$ is a mixed 3-supergraph of $T$. Let $\omega=|w|$. By \Cref{drzftugziuhi}, in time $f(\epsilon)\omega^{O(1)}$, we can compute a  set $\mathcal{C}_U$ of $f(\epsilon)$ mappings $c:X\rightarrow [\omega]_0$ such that for any mapping $c':X\rightarrow [\omega]_0$, there exists come $c \in \mathcal{C}_U$ with $|c-c'|\leq \frac{1}{12}\epsilon \omega$. For some $c \in \mathcal{C}_U$, let $w_c:V(U)\rightarrow [\omega]_0$ be the unique function that is consistent with $w$ and $c$. As $U \in \mathcal{U}$ is arboresque, by \Cref{tfzghik}, in polynomial time, we can compute an orientation $\vec{U}_{c}$ of $U$ that maximizes $R(\vec{U}_{c},w_c)$. Let $\vec{T}_{c}$ be the orientation of $T$ inherited from $\vec{U}_{c}$. We now add $\vec{T}_{c}$ to $\mathcal{T}$. We do this for every $U \in \mathcal{U}$ and every weight function $c \in \mathcal{C}_U$. The final assignment of $\mathcal{T}$ is obtained after the last one of these operations. We will show that $\mathcal{T}$ has the desired properties.
    
    Observe that after the last of these operations, we have that $|\mathcal{T}|=f(\epsilon)$ as $ |\mathcal{U}|\leq 2$ and $|\mathcal{C}_U|=f(\epsilon)$ for all $U \in \mathcal{U}$. Moreover, as $\vec{T}_{c}$ can be computed in polynomial time for every $U \in \mathcal{U}$ and $c \in \mathcal{C}_U$, $ |\mathcal{U}|\leq 2$ and $|\mathcal{C}_U|=O(n^3)$ for all $U \in \mathcal{U}$, we obtain that $\mathcal{T}$ can be computed in polynomial time.

   It remains to show that $\mathcal{T}$ is an $\epsilon$-optimal replacement set. Consider an orientation $\vec{G}$ of $G$. As $\mathcal{U}$ is a simulation set for $((G,w),T)$, there exists some $U \in \mathcal{U}$ and a mapping $w':V(U)\rightarrow [n]_0$ with $|w'|\leq |w|$ that is consistent with $w$ such that for every orientation $\vec{U}$ of $U$, we have that $\vec{G}\langle\vec{U}\rangle$ is an orientation of $G$, $\vec{U}\langle\vec{G}\rangle$ is an orientation of $U$, and $R(\vec{G}\langle\vec{U}\rangle,w)- R(\vec{G},w)\geq R(\vec{U},w')- R(\vec{U}\langle\vec{G}\rangle,w')$ holds. Let $c'=w'|X$. By the definition of $\mathcal{C}_U$, there exists some $c \in \mathcal{C}_U$ with $|c-c'|\leq \epsilon \omega$. Observe that $|w_c+w'|\leq 2|w'|+|w_c-w'|=2|w'|+|c-c'|\leq (2+\epsilon)\omega\leq 3\omega $.
    As $\vec{U}_{c}-X=\vec{T}_{c}$, by  the definition of $w'$, by \Cref{tcvzbun}, by the optimality of $\vec{U}_{c}$ as $|w_c-w'|=|c-c'|$, and by the choice of $c$, we obtain
    \begin{align*}
    R(\vec{G},w)-R(\vec{G}\langle\vec{T}_{c}\rangle,w)
    &=R(\vec{G},w)-R(\vec{G}\langle\vec{U}_{c}\rangle,w)\\
    &\geq R(\vec{U}_{c}\langle \vec{G}\rangle,w')-R(\vec{U}_{c},w')\\
     &=R(\vec{U}_{c}\langle \vec{G}\rangle,w')-R(\vec{U}_{c}\langle \vec{G}\rangle,w_c)+R(\vec{U}_{c}\langle \vec{G}\rangle,w_c)\\&-R(\vec{U}_{c},w_c)+R(\vec{U}_{c},w_c)-R(\vec{U}_{c},w')\\
     &\leq 4 |w_c-w'||w_c+w'|\\
     &=12|c-c'|\omega\\
     &\leq \epsilon \omega^2.
    \end{align*}

As $\vec{T}_{c}\in \mathcal{T}$, we obtain that $\mathcal{T}$ has the desired properties.
\end{proof}
 
\subsection{XP Algorithm}\label{sec44}
This section is dedicated to concluding \Cref{main2}. We first prove \Cref{utgzhu}. After, we conclude a version of \Cref{main2} for conneccted, dismembered instances of WAMMRO, which directly implies \Cref{main2}.

We now prove \Cref{utgzhu}, which we restate here for convenience.
\utgzhu*
 \begin{proof}
     Let $\vec{G}$ be an optimal orientation for $(G,w)$. We now define a sequence $\vec{G}_0,\vec{G}_1,\ldots,\vec{G}_{q}$ of orientations of $G$. First, we set $\vec{G}_0=\vec{G}$. We now recursively define $\vec{G}_{i}$ for $i \in [q]$. Let $i \in [q]$ and suppose that $\vec{G}_{i-1}$ has already been defined. As $\mathcal{T}_i$ is an optimal replacement set for $T_i$, there exists some $\vec{T}_i\in \mathcal{T}_i$ such that $R(\vec{G}_{i-1}\langle\vec{T}_i\rangle,w) \geq R(\vec{G}_{i-1},w)$. Let $\vec{G}_{i}=\vec{G}_{i-1}\langle\vec{T}_i\rangle$.  This finishes the definition of $\vec{G}_0,\vec{G}_1,\ldots,\vec{G}_{q}$. Observe that by the definition of $\vec{G}_0,\vec{G}_1,\ldots,\vec{G}_{q}$ and as $\vec{G}$ is an optimal orientation for $(G,w)$, we have that $\vec{G}_q$ is an optimal orientation for $(G,w)$. Further observe that $\vec{G_q}[V(T_{i})]=\vec{G_i}[V(T_{i})]=\vec{T_i}$ for $i \in [q]$. It follows that $\vec{G_q}$ has the desired properties.
 \end{proof}
Before we prove a version of \Cref{main2} for connected, dismembered instances, we need the following simple observation, which will be reused in \Cref{sec45}.
\begin{proposition}\label{srdzftuzguh}
    Let $(G,w)$ be a connected instance of WAMMRO and let $T_1,\ldots,T_q$ be the undirected components of $G$. Then $q \leq \max\{1,2k\}$. 
\end{proposition}
\begin{proof}
If $k=0$, we obtain that $q=1=\max\{1,2k\}$ as $UG(G)$ is connected and $UG(G)=G$. 

Otherwise, observe that, as $UG(G)$ is connected, for $i \in [q]$, we have that $V(A(G))\cap V(T_i)\neq \emptyset$. This yields $q \leq |V(A(G))|\leq 2k=\max\{1,2k\}$.
\end{proof}
We are now ready to conclude the following version of \Cref{main2} for connected, dismembered instances of WAMMRO.
 \begin{lemma}\label{estdrzftugzuh}
     Let $(G,w)$ be a connected, dismembered instance of WAMMRO. Then an optimal orientation for $(G,w)$ can be computed in time $n^{O(k)}$.
 \end{lemma}
 \begin{proof}Let $T_1,\ldots,T_q$ be the undirected components of $G$.
  If $k=0$, it follows directly from Lemma \ref{tfzghik} that an  optimal orientation for $(G,w)$ can be computed in polynomial time. We may hence suppose that $k \geq 1$.
     It follows by Proposition \ref{srdzftuzguh} that $q \leq 2k$. By \Cref{rdztghok}, in polynomial time, for $i \in [q]$, we can compute an optimal replacement set $\mathcal{T}_i$ for $((G,w),T_i)$ of size $O(n^3)$. By \Cref{utgzhu}, there exists an optimal orientation $\vec{G}$ for $(G,w)$ such that $\vec{G}[V(T_i)]\in \mathcal{T}_i$ for $i \in [q]$. For $i \in [q]$, we let $\vec{T}^1_{i},\ldots,\vec{T}^{\mu_i}_{i}$ be an arbitrary enumeration of $\mathcal{T}_{i}$. By the assumption on $\mathcal{T}_{i}$, we obtain that $\mu_i\leq \alpha n^3$ for some fixed constant $\alpha$. Now consider a vector $x \in [\mu_1]\times\ldots \times [\mu_{q}]$. We define the orientation $\vec{G}_x$ of $G$ by assigning every $e \in E(T_{i})$ the orientation it has in $\vec{T^{x_i}_{i}}$ for $i \in [q]$. Observe that this indeed defines an orientation of $G$. Further, by assumption, we have that there exists some $x^* \in [\mu_1]\times\ldots \times [\mu_{q}]$ such that $\vec{G_{x^*}}=\vec{G}$.

     Our algorithm consists of computing $\vec{G_x}$ for all $x \in [\mu_1]\times\ldots \times [\mu_{q}]$ and maintaining an orientation $\vec{G}_0$ that maximizes $R(\vec{G_0},w)$ among all those found during this procedure. As there exists some $x^* \in [\mu_1]\times\ldots \times [\mu_{q}]$ such that $\vec{G_{x^*}}=\vec{G}$, we obtain that $R(\vec{G_0},w)\geq R(\vec{G},w)$ and hence by the choice of $\vec{G}$, we obtain that $\vec{G_0}$ is an optimal orientation for $(G,w)$. 

     For the running time of the algorithm, observe that $|[\mu_1]\times\ldots \times [\mu_{q}]|=\prod_{i \in [q]}\mu_i\leq (\alpha n^3)^{q}\leq (\alpha n^3)^{2k}=n^{O(k)}$. Further, for every $x \in [\mu_1]\times\ldots \times [\mu_{q}]$, we clearly can compute $\vec{G}_x$ and $R(\vec{G_x},w)$ in polynomial time. Hence the desired running time follows.
 \end{proof}

Now \Cref{main2} follows directly from \Cref{estdrzftugzuh}, \Cref{reddismem}, and \Cref{addwa}.
 \subsection{FPT approximation scheme}\label{sec45}
 This section is dedicated to proving \Cref{main3}. While the remaining parts of the proof are similar to the proof of \Cref{main2}, the main additional difficulty is that applying these techniques, first a solution is obtained whose objective value is an additive amount away from an optimal solution rather than a multiplicative one. In order to overcome this problem, we need to show that all connected instances of WAMMRO admit orientations whose objective value is above a certain threshold. Then we can conclude that a good additive approximation is also a good multiplicative approximation.
 This can be subsumed in the following result.
  \begin{restatable}{lemma}{reciht}\label{reciht}
 Let $\epsilon>0$, let $(G,w)$ be a connected instance of WAMMRO with  $k \geq 1$ and let $\vec{G}_0$ be an orientation of $G$ such that $R(\vec{G}_0,w)\geq R(\vec{G},w)-\frac{\epsilon}{196k^2}|w|^2$ holds for all orientations $\vec{G}$ of $G$. Then $\vec{G}_0$ is a $(1-\epsilon)$-optimal orientation for $(G,w)$.
  \end{restatable}

 The remainder of \Cref{sec45} is structured as follows. In \Cref{sec451}, we prove \Cref{reciht}. After, in \Cref{sec452}, we complete the proof of \Cref{main3}.
 \subsubsection{Lower bound}\label{sec451}

  This section is dedicated to the proof of \Cref{reciht}. We analyze the objective values of some orientations closely related to the ones found in \cite{Hakimi1997OrientingGT} when solving the problem for undirected graphs.

 We first need some more definitions.
 Given an undirected graph $G$, we use $\mathcal{K}(G)$ to denote the set of components of $G$.
 Given a tree together with a weight function $w:V(T)\rightarrow\mathbb{Z}_{\geq 0}$, a {\it centroid} is a vertex $v \in V(T)$ such that $\max\{w(K):K \in \mathcal{K}(T\setminus \{v\})\}$ is minimized. The following result was first proved by Kang and Ault \cite{Kang1975SomePO} for unweighted trees. The proof generalizes verbatim to the weighted version.
 \begin{proposition}\label{tfgzhi}
     Let $(T,w)$ be a weighted tree and $v$ a centroid of $(T,w)$. Then $w(K)\leq \frac{1}{2}w(T)$ for every $K \in \mathcal{K}(T\setminus \{v\})$.
 \end{proposition}
 We next need one more simple result related to the Subset Sum Problem. It is certainly folklore, but proven here for the sake of completeness.
 \begin{proposition}\label{7ftgz8ihuiojo}
     Let $S$ be a multiset of nonnegative integers such that $s \leq \frac{2}{3}\sum S$ holds for all $s \in S$. Then there exists a partition $(S_1,S_2)$ of $S$ such that $\min \{\sum S_1,\sum S_2\}\geq \frac{1}{3}\sum S$.
 \end{proposition}
\begin{proof}
    If there exists some $s \in S$ with $s\geq \frac{1}{3}\sum S$, then $(\{s\},S\setminus \{s\})$ is a partition with the desired properties. We may hence suppose that $s< \frac{1}{3}\sum S$ holds for all $s \in S$. Now let $(S_1,S_2)$ be a partition of $S$ that minimizes $|\sum S_1- \sum S_2|$. Without loss of generality, we may suppose that $\sum S_1\leq \sum S_2$. Suppose for the sake of a contradiction that $\sum S_1<\frac{1}{3}\sum S$. Further, let $s \in S_2$ with $s>0$, let $S_1'=S_1\cup \{s\}$ and $S_2'=S_2\setminus \{s\}$. Clearly, we have $\sum S_2'-\sum S_1'=\sum S_2-\sum S_1-2s<\sum S_2-\sum S_1=|\sum S_1- \sum S_2|$. Further, as $\sum S_1\leq \frac{1}{3}\sum S$ and $s\leq \frac{1}{3}\sum S$, we have 
    \begin{align*}
        \sum S_1'-\sum S_2'&\leq \sum S_1-\sum S_2+2s\\
        &=2 \sum S_1+2s-\sum S\\
        &< \frac{1}{3}\sum |S|\\
        &= (\sum S-\frac{1}{3}\sum S)-\frac{1}{3}\sum S\\
        &< \sum S_2-\sum S_1\\
        &=|\sum S_1- \sum S_2|.
    \end{align*}
    We hence obtain a contradiction to the choice of $(S_1,S_2)$. We obtain that $\min \{\sum S_1,\sum S_2\}=\sum S_1\geq \frac{1}{3}\sum S$.
\end{proof}
We are now ready to prove that every instance $(G,w)$ of WAMMRO in which $G$ is undirected admits an orientation whose objective value is above a certain threshold.
 \begin{lemma}\label{esrzdtufziguohi}
     Let $(G,w)$ be a connected instance of WAMMRO such that $G$ is undirected with $|w|\geq 2$. Then there exists an orientation $\vec{G}$ of $G$ such that $R(\vec{G},w)\geq \frac{1}{49}|w|^2$.
 \end{lemma}
 \begin{proof}Observe that $G$ is a tree as $(G,w)$ is a connected instance of WAMMRO.
     First suppose that $|w| \leq 7$. If there exists some $v \in V(G)$ with $w(v)\geq 2$, then for an arbitrary orientation $\vec{G}$ of $G$, we have $R(\vec{G},w)\geq 2{w(v)\choose 2}\geq 2\geq \frac{1}{49}|w|^2$. Otherwise, as $|w|\geq 2$, there exist two vertices $u,v \in V(G)$ with $\min\{w(u),w(v)\}\geq 1$. As $G$ is connected, we can choose an orientation $\vec{G}$ of $G$ that satisfies $\kappa_{\vec{G}}(u,v)=1$. We obtain $R(\vec{G},w)\geq \kappa_{\vec{G}}(u,v)w(u)w(v)\geq 1\geq \frac{1}{49}|w|^2$.

     We may hence suppose that $|w| \geq 8$. Let $v$ be a centroid of $(G,w)$. If $w(v)\geq \frac{1}{4}|w|$, then, as $|w| \geq 8$, for an arbitrary orientation $\vec{G}$ of $G$, we have $R(\vec{G},w)\geq 2{w(v)\choose 2}\geq \frac{1}{4}|w|(\frac{1}{4}|w|-1)\geq \frac{1}{4}|w|\frac{1}{8}|w|\geq\frac{1}{49}|w|^2$.

     We may hence suppose that $w(v)\leq \frac{1}{4}|w|$. We then have $w(\mathcal{K}(G-\{v\}))=|w|-w(v)\geq \frac{3}{4}|w|$. Further, by Proposition \ref{tfgzhi}, we have $w(K)\leq \frac{1}{2}|w|\leq \frac{2}{3}w(\mathcal{K}(G-\{v\}))$ for every $K \in \mathcal{K}(G-\{v\})$. It hence follows by Proposition \ref{7ftgz8ihuiojo} that there exists a partition $(\mathcal{K}_1,\mathcal{K}_2)$ of $\mathcal{K}(G-\{v\})$ such that $\min \{w(\mathcal{K}_1),w(\mathcal{K}_2)\}\geq \frac{1}{3}w(\mathcal{K}(G-\{v\}))\geq \frac{1}{4}|w|$. We now choose an orientation $\vec{G}$ of $G$ such that for every $K \in \mathcal{K}_1$ and every $u \in V(K)$, there exists a directed path from $u$ to $v$ in $\vec{G}$ and for every $K \in \mathcal{K}_2$ and every $u \in V(K)$, there exists a directed path from $v$ to $u$ in $\vec{G}$. Observe that $\kappa_{\vec{G}}(u,u')=1$ for all $u \in V(\mathcal{K}_1)$ and $u' \in V(\mathcal{K}_2)$. This yields \begin{align*}        
    R(\vec{G},w)&\geq \sum_{(u,u')\in P(G)}\kappa_{\vec{G}}(u,u')w(u)w(u')\\
    &\geq \sum_{u \in V(\mathcal{K}_1)}w(u)\sum_{u' \in V(\mathcal{K}_2)}w(u')\kappa_{\vec{G}}(u,u')\\
    &=\sum_{u \in V(\mathcal{K}_1)}w(u)\sum_{u' \in V(\mathcal{K}_2)}w(u')\\
    &=w(\mathcal{K}_1)w(\mathcal{K}_2)\\
    &\geq \frac{1}{16}|w|^2\\
    &\geq \frac{1}{49}|w|^2.
    \end{align*}
 \end{proof}
We are now able to generalize this result to arbitrary connected instances $(G,w)$ of WAMMRO by considering a largest undirected component of $G$. 
 \begin{lemma}\label{asefhjikj}
     Let $(G,w)$ be a connected instance of WAMMRO with  $k \geq 1$ that admits an orientation $\vec{G}$ with $R(\vec{G},w)\geq 1$. Then there exists an orientation $\vec{G}$ of $G$ with $R(\vec{G},w)\geq \frac{1}{196k^2}|w|^2$.
 \end{lemma}
 \begin{proof}
First suppose that $|w|\leq 2k$. Then we set $\vec{G}$ to be an orientation of $G$ that satisfies $R(\vec{G},w)\geq 1$. This yields $R(\vec{G},w)\geq 1 \geq (\frac{|w|}{2k})^2\geq \frac{1}{196k^2}|w|^2$.

We may hence suppose that $|w| \geq 2k+1$. Let $T_1,\ldots,T_q$ be the undirected components of $G$. By Proposition \ref{srdzftuzguh}, we have $q \leq 2k$. Observe that $\sum_{i \in [q]}w(T_i)=|w|$. By symmetry, we may hence suppose that $w(T_1)\geq \lceil\frac{1}{q}|w|\rceil\geq \lceil\frac{1}{2k}|w|\rceil$. In particular, it follows that $w(T_1)\geq 2$. We hence obtain by Lemma \ref{esrzdtufziguohi} that there exists an orientation $\vec{T_1}$ of $T_1$ with $R(\vec{T_1},w|V(T_1))\geq \frac{1}{49}(w(T_1))^2\geq \frac{1}{196k^2}|w|^2$. Now let $\vec{G}$ be an arbitrary orientation of $G$ in which all edges of $E(T_1)$ have the same orientation as in $\vec{T_1}$. Then we have $R(\vec{G},w)\geq R(\vec{T_1},w|V(T_1))\geq \frac{1}{196k^2}|w|^2$.
 \end{proof}

We now ready to give the proof of \Cref{reciht}, which we restate here for convenience.
\reciht*
\begin{proof}
     Let $\vec{G}$ be an orientation of $G$. We may suppose that $\vec{G}$ is optimal. If $R(\vec{G},w)=0$, then every orientation for $(G,w)$ is optimal, so the statement trivially holds. We may hence suppose that $R(\vec{G},w)\geq 1$. By Lemma \ref{asefhjikj}, we obtain that $R(\vec{G},w)\geq \frac{1}{196k^2}|w|^2$. By assumption, this yields
     \begin{align*}
         R(\vec{G_0},w)&\geq R(\vec{G},w)-\frac{\epsilon}{196k^2}|w|^2\\
         &\geq R(\vec{G},w)-\epsilon R(\vec{G},w)\\
         &=(1-\epsilon)R(\vec{G},w).
     \end{align*}
 \end{proof}
 \subsubsection{Main proof of Theorem \ref{main3}}\label{sec452}
This section is dedicated to completing the proof of \Cref{main3}. We make use of \Cref{reciht} and otherwise use similar arguments as in \Cref{sec44} for the proof of \Cref{main2}. The following result is the desired analogue of \Cref{utgzhu}.
 \begin{lemma}\label{gvfziuog}
     Let $\epsilon>0$ and let $(G,w)$ be a connected, dismembered instance of WAMMRO with $k\geq 1$. Further, let $\{T^1,\ldots,T^{q}\}$ be the set of  undirected components of $G$ and for $i \in [q]$, let $\mathcal{T}_i$ be an $\frac{\epsilon}{392k^3}$-optimal replacement set for $T_i$. Then there exists a $(1-\epsilon)$-optimal orientation $\vec{G}_1$ for $(G,w)$ with $\vec{G_1}[V(T^{i})]\in \mathcal{T}^{i}$ for $i \in [q]$.
 \end{lemma}
 \begin{proof}

     Let $\vec{G}$ be an orientation of $G$. We may suppose that $\vec{G}$ is optimal for $(G,w)$. We define a sequence $\vec{G}_0,\vec{G}_1,\ldots,\vec{G}_{q}$ of orientations of $G$. First, we set $\vec{G}_0=\vec{G}$. We now recursively define $\vec{G}_{i}$ for $i \in [q]$. Let $i \geq 1$ and suppose that $\vec{G}_{i-1}$ has already been constructed. By the definition of $\mathcal{T}_{i}$, there exists some $\vec{T}_{i}\in \mathcal{T}_{i}$ such that $R(\vec{G}_{i-1}\langle \vec{T}_{i}\rangle,w) \geq R(\vec{G}_{i-1},w)-\frac{\epsilon}{392k^3}|w|^2$. Let $\vec{G}_{i}=\vec{G}_{i-1}\langle \vec{T}_{i}\rangle$.  This finishes the definition of $\vec{G}_0,\vec{G}_1,\ldots,\vec{G}_{q}$.
     
     We will show that $\vec{G}_q$ has the desired properties. 
     Observe that by the definition of $\vec{G}_0,\vec{G}_1,\ldots,\vec{G}_{q}$, we have that 
     $R(\vec{G}_{i},w)\geq R(\vec{G}_{i-1},w)-\frac{\epsilon}{392k^3}|w|^2$ for $i \in [q]$. As $q \leq 2k$ by Proposition \ref{srdzftuzguh} and $k\geq 1$, we obtain that $R(\vec{G}_q,w)\geq R(\vec{G},w)-\frac{\epsilon}{196k^2}|w|^2$. It hence follows from \Cref{reciht} that $\vec{G}_q$ is a $(1-\epsilon)$-optimal orientation for $(G,w)$. Further observe that $\vec{G}_q[V(T^{i})]=\vec{G}_{i}[V(T^{i})]=\vec{T}_{i}$ for $i \in [q]$. 
 \end{proof}
We are now ready to prove a version of \Cref{main3} for connected, dismembered instances of WAMMRO. This result is an analogue of \Cref{estdrzftugzuh}. 
 \begin{lemma}\label{dwqdeqwd}
     Let $\epsilon>0$ and let $(G,w)$ be a connected, dismembered instance of WAMMRO. Then a $(1-\epsilon)$-optimal orientation $\vec{G}$ for $(G,w)$ can be computed in time $f(k,\epsilon)n^{O(1)}$.
 \end{lemma}
 \begin{proof}
 If $k=0$, it follows directly from Lemma \ref{tfzghik} that an  optimal orientation for $(G,w)$ can be computed in polynomial time. We may hence suppose that $k \geq 1$.
     Let $T_1,\ldots,T_{q}$ be the  undirected components of $G$. By \Cref{rdztghgsfok}, for $i \in [q]$, in polynomial time, we can compute an $\frac{\epsilon}{392k^3}$-optimal replacement set $\mathcal{T}_{i}$ for $((G,w),T_i)$ of size $f(k,\epsilon)$. By \Cref{gvfziuog}, there exists a $(1-\epsilon)$-optimal orientation $\vec{G}$ for $(G,w)$ with $\vec{G}[V(T_{i})]\in \mathcal{T}_{i}$ for $i \in [q]$. For $i \in [q]$, we let $\vec{T}^1_{i},\ldots,\vec{T}^{\mu_i}_{i}$ be an arbitrary enumeration of $\mathcal{T}_{i}$. By the assumption on $\mathcal{T}_{i}$, we obtain that $\mu_i=f(k,\epsilon)$. Now consider a vector $x \in [\mu_1]\times\ldots \times [\mu_{q}]$. We define the orientation $\vec{G}_x$ of $G$ by assigning every $e \in E(T_{i})$ the orientation it has in $\vec{T}^{x_i}_{i}$ for $i \in [q]$. Observe that this indeed defines an orientation of $G$. Further, by assumption, we have that there exists some $x^* \in [\mu_1]\times\ldots \times [\mu_{q}]$ such that $\vec{G_{x^*}}=\vec{G}$.

     Our algorithm consists of computing $\vec{G}_x$ and $R(\vec{G_x},w)$ for all $x \in [\mu_1]\times\ldots \times [\mu_{q}]$ and outputting an orientation $\vec{G_0}$ that maximizes $R(\vec{G_0},w)$ among all those found during this procedure. As there exists some $x^* \in [\mu_1]\times\ldots \times [\mu_{q}]$ such that $\vec{G_{x^*}}=\vec{G}$, we obtain that $R(\vec{G}_0,w)\geq R(\vec{G},w)$ and hence, by the choice of $\vec{G}$, we obtain that $\vec{G_0}$ is a $(1-\epsilon)$-optimal orientation for $(G,w)$. 

     For the running time of the algorithm, observe that $|[\mu_1]\times\ldots \times [\mu_{q}]|=\prod_{i \in [q]}\mu_i$. As $q\leq 2k$ by Proposition \ref{srdzftuzguh} and $k \geq 1$, and $\mu_i=f(k,\epsilon)$ for $i \in [q]$, we obtain that $|[\mu_1]\times\ldots \times [\mu_{q}]|=f(k,\epsilon)$.  Further, for every $x \in [\mu_1]\times\ldots \times [\mu_{q}]$, we clearly can compute $\vec{G}_x$ and $R(\vec{G}_x,w)$ in polynomial time. Hence the desired running time follows.
 \end{proof}

Now \Cref{main3} follows directly from \Cref{dwqdeqwd}, \Cref{reddismem}, and \Cref{addwa}.
 \section{Conclusion}\label{sec5}
In this article, we considered the problem of finding an orientation of a mixed graph $G$ that maximizes the total number of vertex pairs such that the second vertex in the pair is reachable from the first
 one. We first showed in Theorem \ref{main1} that the maximization problem is APX-hard. Then, we considered the parameterization by the number $k$ of arcs in $A(G)$. We show in Theorems \ref{main2} and \ref{main3} that the problem is XP with respect to $k$ and that an arbitrarily good approximation can be computed in FPT time with respect to $k$.

 We hope for these results to be the starting point for more research on the complexity of MMRO. There are several interesting follow-up questions that seem worthwhile studying. The perhaps most interesting one, in particular in the light of Theorem \ref{main1}, is whether a constant-factor approximation algorithm is available.

 \begin{question}\label{rxstzdtufziguhijo} 
     Is there a constant $\alpha>0$ such that there exists a polynomial-time $\alpha$-approximation algorithm for MMRO?
 \end{question}
  We wish to remark that the corresponding questions for FPUMRO also seems to be open. 

 One approach for Question \ref{rxstzdtufziguhijo} is to determine whether for any instance $(G,w)$ of WAMMRO, there exists an orientation of $G$ that maintains a constant fraction of the reachabilities in $G$ with respect to the given weight function.

 \begin{question}\label{etxsrdtfzguhi} 
     Does there exist a constant $\alpha>0$ such that for any instance $(G,w)$ of WAMMRO, there exists an orientation $\vec{G}$ for $(G,w)$ such that $R(\vec{G},w)\geq \alpha R(G,w)$?
 \end{question}

 It is easy to see that an affirmative answer to Question \ref{etxsrdtfzguhi} implies an affirmative answer to Question \ref{rxstzdtufziguhijo}.

 Moreover, it is natural to ask about a common strengthening of \Cref{main2} and \Cref{main3}, namely if MMRO is actually FPT when parameterized by $k$.

 \begin{question}\label{serdtfzgui}
     Is there an algorithm that computes an optimal orientation for any instance $G$ of MMRO and runs in time $f(k)n^{O(1)}$?
 \end{question}

 One way to achieve this would be to improve \Cref{rdztghok} in the sense that an optimal replacement set of size only dependent on $k$ is found rather than one of size $O(n^3)$. However, this is not possible as the following result shows.
 
 \begin{proposition}
     For any positive integer $q$, there exists a connected, dismembered instance $(G,w)$ of WAMMRO with $k=2$ and an undirected component $T$ of $G$ such that every optimal replacement set for $((G,T),w)$ is of size at least $q+1$.
 \end{proposition}
 \begin{proof}
     Let a positive integer $q$ be fixed. We now construct an instance $(G,w)$ of WAMMRO in the following way. We let $V(G)$ consist of a set of $2q+4$ vertices $\{x,y,u_0,\ldots,u_q,v_0,\ldots,v_q\}$. We next let $E(G)$ contain edges so that $u_0\ldots u_q$ forms a path $T$ and we let $E(G)$ contain the edge $v_ix$ for $i \in [q]_0$. Finally, we let $A(G)$ consist of the arcs $xu_0$ and $yv_q$. We now define the weight function $w: V(G)\rightarrow \mathbb{Z}_{\geq 0}$. First, let $N=2{q+2 \choose 2}2^{2q^2}$. Now, for $i \in [q]_0$, we set $w(v_i)=2^{2i}N$ and $w(u_i)=2^{q^2-i^2}$. Finally, we set $w(x)=0$ and $w(y)=N$. This finishes the description of $((G,w),T)$. Observe that $(G,w)$ is a dismembered instance of WAMMRO and $T$ is an undirected component of $G$. For an illustration, see Figure \ref{qwe}. 
\begin{figure}[h]
    \centering
        \includegraphics[width=.8\textwidth]{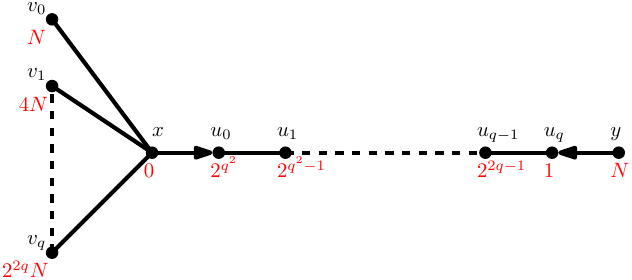}
        \caption{An illustration of the constructed instance $(G,w)$. The vertex names are marked in black and the weight function $w$ is marked in red.}\label{qwe}
\end{figure}

Let $P_0$ be the set of pairs $(a,b)$ in $P(G)$ with $\{a,b\}\subseteq V(T)\cup x$ and $P_1=P(G)\setminus P_0$. For any $(a,b)\in P_0$ we have $w(a)w(b)\leq 2^{2q^2}$. It follows that $\sum_{(a,b)\in P_0}w(a)w(b)\leq |P_0|2^{2q^2}=N$.

     Now for $i \in [q]_0$, let $\vec{T}_i$ be the unique orientation of $T$ in which $u_i$ is reachable from all other vertices in $V(T)$. Further, let $\mathcal{T}$ be an optimal replacement set for $((G,T),w)$. We will show that $\vec{T}_i\in \mathcal{T}$ for $i \in [q]_0$.

     In order to do so, fix some $i \in [q]_0$. Let $H$ be the partial orientation of $G$ which is obtained from $G$ by orienting $v_ix$ from $v_i$ to $x$ and by orienting $v_jx$ from $x$ to $v_j$ for all $j \in [q]_0 \setminus \{i\}$. Let $\vec{H}$ be an optimal orientation for $(H,w)$. As $\mathcal{T}$ is an optimal replacement set for $((G,w),T)$, there exists some $\vec{T}\in \mathcal{T}$ such that $R(\vec{H}\langle \vec{T}\rangle,w) \geq R(\vec{H},w)$. Let $\vec{H}_0=\vec{H}\langle \vec{T}\rangle$. By the choice of $\vec{H}$ and $\vec{H}_0$, we have that $\vec{H}_0$ is an optimal orientation for $(H,w)$.

     Now, let $j_1$ be the largest integer in $[q]_0$ such that $u_{j_1}$ is reachable from $x$ in $\vec{H}_0$. We will show that $j_1 \geq i$. Suppose otherwise. Then, by construction and assumption, we have that in $\vec{H}_0$ the unique edge linking $u_{j_1}$ and $u_{j_1+1}$ is oriented from $u_{j_1+1}$ to $u_{j_1}$. Let $\vec{H}_1$ be obtained from $\vec{H}_0$ by reorienting this edge from $u_{j_1}$ to $u_{j_1+1}$. Observe that $\kappa_{\vec{H}_1}(v_i,u_{j_1+1})>\kappa_{\vec{H}_0}(v_i,u_{j_1+1})$. Moreover, for any pair $(a,b)\in P_1$, we have either $\kappa_{\vec{H}_1}(a,b)\geq \kappa_{\vec{H}_0}(a,b)$ or $(a,b)=(y,u_{j_1})$.

     As $q\geq i>j_1$, this yields that 
     \begin{align*}
         R(\vec{H}_1,w)-R(\vec{H}_0,w)
         &\geq w(v_i)w(u_{j_1+1})-w(y)w(u_{j_1})-\sum_{(a,b)\in P_0}w(a)w(b)\\
         &\geq 2^{2i}N2^{q^2-(j_1+1)^2}-N2^{q^2-j_1^2}-N\\
         &=N(2^{q^2-j_1^2+2i-2j_1-1}-2^{q^2-j_1^2}-1)\\
         &\geq N.
     \end{align*}
This contradicts $\vec{H}_0$ being optimal.

     Now, let $j_2$ be the smallest integer in $[q]_0$ such that $u_{j_2}$ is reachable from $y$ in $\vec{H}_0$. We will show that $j_2 \leq i$. Suppose otherwise. Then, by construction and assumption, we have that in $\vec{H}_0$ the unique edge linking $u_{j_2}$ and $u_{j_2-1}$ is oriented from $u_{j_2-1}$ to $u_{j_2}$. Let $\vec{H}_2$ be obtained from $\vec{H}_0$ by reorienting this edge from $u_{j_2}$ to $u_{j_2-1}$. Observe that $\kappa_{\vec{H}_2}(y,u_{j_2-1})>\kappa_{\vec{H}_0}(y,u_{j_2-1})$. Moreover, for any pair $(a,b)\in P_1$, we have either $\kappa_{\vec{H}_2}(a,b)\geq \kappa_{\vec{H}_0}(a,b)$ or $(a,b)=(v_i,u_{j_2})$.

     As $q\geq j_2>i$, this yields that 
     \begin{align*}
         R(\vec{H}_2,w)-R(\vec{H}_0,w)
         &\geq w(y)w(u_{j_2-1})-w(v_i)w(u_{j_2})-\sum_{(a,b)\in P_0}w(a)w(b)\\
         &\geq N2^{q^2-(j_2-1)^2}-2^{2i}N2^{q^2-j_2^2}-N\\
         &=N(2^{q^2-j_2^2+2j_2-1}-2^{q^2-j_2^2+2i}-1)\\
         &\geq N.
     \end{align*}
This contradicts $\vec{H}_0$ being optimal.

It follows that $u_i$ is reachable from both $x$ and $y$ in $\vec{H}_0$. Hence, by construction, we obtain that $\vec{H}_0[V(T)]=\vec{T}_i$. By the definition of $\vec{H}_0$, it follows that $\vec{T}=\vec{T}_i$, so in particular $\vec{T}_i \in \mathcal{T}$. As $i \in [q]_0$ was chosen arbitrarily, it follows that $|\mathcal{T}|\geq q+1$.
 \end{proof}
 
 A further interesting question is how structural parameters influence the difficulty of MMRO. In particular, it would be good to understand how, given an instance $(G,w)$ of WAMMRO, the treewidth $tw$ of $UG(G)$ influences the difficulty of the problem.  As every undirected component of $G$ is a tree, we obtain that $tw \leq k+1$. It follows that an XP algorithm for $tw$ would imply a qualitative version of \Cref{main2}. While it seems plausible that such an algorithm can be obtained by a dynamic programming approach, we believe that it would need to be significantly more technical than the algorithm presented in \Cref{sec42}. Similarly, an FPT algorithm for $tw$ would yield an affirmative answer for \Cref{serdtfzgui}.
 
 On a higher level, except for \cite{hanaka2025structuralparameterssteinerorientation}, to our best knowledge, this is the first time that the number of arcs has been considered as a parameter for orientation problems of mixed graphs. It would be interesting to consider this parameterization for other problems on mixed graphs. For this question to be meaningful, the problem needs to be solvable in undirected graphs, but hard in mixed graphs.
 
 One candidate is the problem of finding a 2-vertex-connected orientation of a mixed graph $G$. This problem can be solved in polynomial time if $G$ is an undirected graph \cite{THOMASSEN201567} and is NP-hard in general mixed graphs \cite{HORSCH2023100774}. It would hence be interesting to see whether this problem is XP or even FPT with respect to the number of arcs of $G$.

 Another candidate is finding an orientation of a mixed graph that is a well-balanced orientation of its underlying graph. For the definition of a well-balanced orientation, see \cite{10.1016/j.dam.2011.09.003}. The strong orientation theorem of Nash-Williams states that every undirected graph has a well-balanced orientation \cite{NashWilliams1960OnOC}. However, the above described orientation problem in mixed graphs is NP-hard \cite{Bernth2008WellbalancedOO}. Again, the parameterization by the number of arcs could be considered.

\bibliography{lipics-v2021-sample-article}

\end{document}